\documentclass[a4paper,UKenglish,cleveref, autoref]{lipics-v2021}

\usepackage{amsmath}
\usepackage{amsthm}
\usepackage{graphicx}
\usepackage[colorinlistoftodos]{todonotes}
\usepackage{algorithm}
\usepackage[noend]{algpseudocode}
\usepackage{enumerate}
\usepackage{xcolor}
\usepackage{framed}
\usepackage{bm}
\usepackage{pifont}
\usepackage{multicol}
\usepackage{lipsum}
\usepackage{tikz}
\usetikzlibrary{calc}
\usepackage{wrapfig}
\usepackage{array}
\usepackage{comment}
\usepackage{url}
\usepackage{caption}
\usepackage{subcaption}
\usepackage{xstring}
\usepackage{varwidth}
\usepackage{sansmath}
 \usepackage{xspace}
\usepackage{mathtools}

\makeatletter
\newtheorem*{rep@theorem}{\rep@title}
\newcommand{\newreptheorem}[2]{%
\newenvironment{rep#1}[1]{%
 \def\rep@title{#2 \ref{##1}}%
 \begin{rep@theorem}}%
 {\end{rep@theorem}}}
\makeatother

\newreptheorem{theorem}{Theorem}
\newreptheorem{claim}{Claim}

\algnewcommand{\LeftComment}[1]{\Statex \(\triangleright\) #1}

\DeclarePairedDelimiter\floor{\lfloor}{\rfloor}

\algdef{SE}[RECEIVING]{Receiving}{EndReceiving}[1]{\textbf{upon
receiving}\ #1\ \algorithmicdo}{\algorithmicend\ \textbf{}}%
\algtext*{EndReceiving}

\algdef{SE}[UPON]{Upon}{EndUpon}[1]{\textbf{upon
}\ #1\ \algorithmicdo}{\algorithmicend\ \textbf{}}%
\algtext*{EndUpon}

\newcommand{\lr}[1]{\langle #1 \rangle}

\newcommand{\ignore}[1]{}

\newcommand{\lv}[1]{$\mathsf{startView(#1)}$\xspace}
\newcommand{\wv}[1]{$\mathsf{wedgeView(#1)}$\xspace}
\newcommand{\barriersync}[1]{$\mathsf{barrier\text{-}sync(#1)}$\xspace}
\newcommand{\barrierready}[1]{$\mathsf{barrier\text{-}ready(#1)}$\xspace}
\newcommand{\elect}[1]{$\mathsf{elect(#1)}$\xspace}
\newcommand{\we}[1]{$\mathsf{exchangeState(#1)}$\xspace}

\def\HiLi{\leavevmode\rlap{\hbox to
\hsize{\color{gray!10}\leaders\hrule height .8\baselineskip
depth .5ex\hfill}}}

\algblockdefx[ON]{On}{EndOn}[1]
  {{\bf On} #1\   \algorithmicdo }
  {\algorithmicend\ }
\makeatletter
\ifthenelse{\equal{\ALG@noend}{t}}%
  {\algtext*{EndOn}}
  {}%
\makeatother

\makeatletter

\tikzset{%
  remember picture with id/.style={%
    remember picture,
    overlay,
    save picture id=#1,
  },
  save picture id/.code={%
    \edef\pgf@temp{#1}%
    \immediate\write\pgfutil@auxout{%
      \noexpand\savepointas{\pgf@temp}{\pgfpictureid}}%
  },
  if picture id/.code args={#1#2#3}{%
    \@ifundefined{save@pt@#1}{%
      \pgfkeysalso{#3}%
    }{
      \pgfkeysalso{#2}%
    }
  }
}

\def\savepointas#1#2{%
  \expandafter\gdef\csname save@pt@#1\endcsname{#2}%
}

\def\tmk@labeldef#1,#2\@nil{%
  \def\tmk@label{#1}%
  \def\tmk@def{#2}%
}

\tikzdeclarecoordinatesystem{pic}{%
  \pgfutil@in@,{#1}%
  \ifpgfutil@in@%
    \tmk@labeldef#1\@nil
  \else
    \tmk@labeldef#1,(0pt,0pt)\@nil
  \fi
  \@ifundefined{save@pt@\tmk@label}{%
    \tikz@scan@one@point\pgfutil@firstofone\tmk@def
  }{%
  \pgfsys@getposition{\csname save@pt@\tmk@label\endcsname}\save@orig@pic%
  \pgfsys@getposition{\pgfpictureid}\save@this@pic%
  \pgf@process{\pgfpointorigin\save@this@pic}%
  \pgf@xa=\pgf@x
  \pgf@ya=\pgf@y
  \pgf@process{\pgfpointorigin\save@orig@pic}%
  \advance\pgf@x by -\pgf@xa
  \advance\pgf@y by -\pgf@ya
  }%
}

\newlength\AlgIndent
\newcounter{mymark}
\newcommand\ColorLine{%
  \stepcounter{mymark}%
  \tikz[remember picture with id=mark-\themymark,overlay] {;}%
  \begin{tikzpicture}[remember picture,overlay]%
    \filldraw[gray!20]%
   let \p1=(pic cs:mark-\themymark), 
   \p2=(current page.east)  in 
   ([xshift=-\ALG@thistlm-0em,yshift=-0.7ex]0,\y1)  rectangle
   ++(\linewidth+\AlgIndent,\baselineskip); \end{tikzpicture}%
}%

\newcommand\ColorLinex{%
  \stepcounter{mymark}%
  \tikz[remember picture with id=mark-\themymark,overlay] {;}%
  \begin{tikzpicture}[remember picture,overlay]%
    \filldraw[gray!20]%
   let \p1=(pic cs:mark-\themymark), 
   \p2=(current page.east)  in 
   ([xshift=-\ALG@thistlm--3em,yshift=-0.7ex]0,\y1) 
   rectangle ++(\linewidth+\AlgIndent,\baselineskip); \end{tikzpicture}%
}%

\algnewcommand\CREQUIRE{\item[\setlength\AlgIndent{1.6em}\ColorLine\algorithmicrequire]}%
\algnewcommand\CENSURE{\item[\setlength\AlgIndent{1.6em}\ColorLine\algorithmicensure]}%
\algnewcommand\CSTATE{\State\ColorLine}%
\algnewcommand\CSTATEx{\Statex\ColorLinex}%
\algnewcommand\CCOMMENT{\Comment\ColorLine}%

\title{In Search for an Optimal Authenticated Byzantine Agreement} 

\titlerunning{In Search for an Optimal Authenticated Byzantine
Agreement} 

\author{Alexander Spiegelman}{Novi
Research}{sasha.spiegelman@gmail.com}{}{}

\authorrunning{Alexander Spiegelman} 

\Copyright{Alexander Spiegelman} 



\category{} 

\relatedversion{} 

\acknowledgements{}

\EventEditors{John Q. Open and Joan R. Access}
\EventNoEds{2}
\EventLongTitle{42nd Conference on Very Important Topics (CVIT 2016)}
\EventShortTitle{CVIT 2016}
\EventAcronym{CVIT}
\EventYear{2016}
\EventDate{December 24--27, 2016}
\EventLocation{Little Whinging, United Kingdom}
\EventLogo{}
\SeriesVolume{42}
\ArticleNo{23}

\begin{document}

\maketitle

\begin{abstract}

In this paper, we challenge the conventional approach of state machine
replication systems to design deterministic agreement
protocols in the eventually synchronous communication model.
We first prove that no such protocol can guarantee bounded
communication cost before the global stabilization time and propose
a different approach that hopes for the best (synchrony) but prepares
for the worst (asynchrony).
Accordingly, we design an \emph{optimistic} byzantine agreement
protocol that first tries an efficient deterministic algorithm that
relies on synchrony for termination only, and then, only if an agreement was not reached due to asynchrony, the protocol uses a
randomized asynchronous protocol for fallback that guarantees
termination with probability $1$.

We formally prove that our protocol achieves optimal communication
complexity under all network conditions and failure scenarios. 
We first prove a lower bound of $\Omega(ft+ t)$ for synchronous
deterministic byzantine agreement protocols, where $t$ is the failure
threshold, and $f$ is the actual number of failures. Then, we present
a tight upper bound and use it for the synchronous part of the
optimistic protocol.
Finally, for the asynchronous fallback, we use a variant of the
(optimal) VABA protocol, which we reconstruct to safely combine it with
the synchronous part.

We believe that our adaptive to failures synchronous byzantine
agreement protocol has an independent interest since it is the first protocol we are aware of which communication complexity optimally
depends on the actual number of failures.

\end{abstract}

 \section{Introduction}
\label{sec:intro}

With the emergence of the Blockchain use case, designing
efficient geo-replicated Byzantine tolerant state machine replication (SMR) systems is now one of the most challenging problems in distributed computing.
The core of every Byzantine SMR system is
the Byzantine agreement problem (see~\cite{bano2019sok} for a
survey), which was first introduced four decades ago~\cite{PSL80}
and has been intensively studied since then~\cite{
CachinSecure, king2016byzantine, malkhi2019flexible}.
The bottleneck in geo-replicated SMR systems is
the network communication, and thus a substantial effort in recent
years was invested in the search for an optimal communication Byzantine
agreement protocol~\cite{gueta2019sbft, hotstuff, tendermint,
naor2019cogsworth}. 

To circumvent the FLP~\cite{FLP85} result that states that
deterministic asynchronous agreement protocols are impossible,
most SMR solutions~\cite{pbft,
gueta2019sbft, hotstuff, zyzzyva} assume eventually
synchronous communication models and provide safety during asynchronous periods
but can guarantee progress only after the global
stabilization time (GST).

Therefore, it is quite natural that state-of-the-art authenticated
Byzantine agreement protocols~\cite{gueta2019sbft, hotstuff,
tendermint, naor2019cogsworth} focus on reducing communication cost
after GST, while putting up with the potentially unbounded cost
beforehand.
For example, Zyzzyva~\cite{zyzzyva} and later
SBFT~\cite{gueta2019sbft} use threshold
signatures~\cite{shoup2000practical} and collectors to reduce the
quadratic cost induced by the all-to-all communication in each view
of the PBFT~\cite{pbft} protocol.
HotStuff~\cite{hotstuff} leverages ideas presented in
Tendermint~\cite{tendermint} to propose a linear view-change
mechanism, and a few follow-up works~\cite{naor2019cogsworth,
naor2020expected, bravo2020making} proposed algorithms for
synchronizing parties between views.
Some~\cite{naor2019cogsworth, naor2020expected} proposed a synchronizer    
with a linear cost after GST in failure-free runs, while 
others~\cite{bravo2020making} provided an implementation that
guarantees bounded memory even before GST. 
However, none of the above algorithms bounds the number of 
views executed before GST, and thus none of them can guarantee a
bounded total communication cost.

We argue in this paper that designing agreement algorithms in the eventually synchronous model is not the best approach to reduce the total communication complexity of SMR systems
and propose an alternative approach.
That is, we propose to forgo the eventually synchronous assumptions and instead optimistically consider the network to be synchronous and immediately
switch to randomized asynchronous treatment if synchrony assumption
does not hold.
Our goal in this paper is to develop an \emph{optimistic} protocol
that adapts to network conditions and actual failures to
guarantee termination with an optimal communication cost under all
failure and network scenarios.

\subsection{Contribution}

\textbf{Vulnerability of the eventually synchronous model.} 
A real network consists of synchronous and asynchronous periods.
From a practical point of view, if the synchronous periods are too
short, no deterministic Agreement algorithm can make
progress~\cite{FLP85}.
Therefore, to capture the assumption that eventually there will be a
long enough synchronous period for a deterministic Agreement to
terminate, the eventually synchronous model assumes that every
execution has a point, called GST, after which the network is
synchronous.
In our first result, we capture the inherent vulnerability of
algorithms designed in the eventually synchronous communication model.
That is, we exploit the fact that GST can occur after an arbitrarily long time to prove the following lower
bound:

\begin{theorem}
\label{theorem:ES}

There is no eventually synchronous deterministic Byzantine agreement protocol that can tolerate a single failure and guarantee bounded
communication cost even in failure-free runs.

\end{theorem}

\noindent\textbf{Tight bounds for synchronous Byzantine
agreement.}  
To develop an optimal optimistic protocol that achieves optimal
communication under all failure and network scenarios we first
establish what is the best we can achieve in synchronous
settings.
Dolev and Reischuk~\cite{dolev1985bounds} proved that there is no deterministic protocol that solves synchronous Byzantine agreement with
$o(t^2)$ communication cost, where $t$ is the failure threshold.
We generalize their result by considering the actual number of failures
$f \leq t$ and prove the following lower bound:
\begin{theorem}
\label{theorem:S}

Any synchronous deterministic Byzantine agreement protocol has
$\Omega(ft + t)$ communication complexity.

\end{theorem}
\noindent It is important to note that the lower bound holds even for
deterministic protocols that are allowed to use perfect cryptographic
schemes such as threshold signatures and authenticated links.
Then, we present the first deterministic
cryptography-based synchronous Byzantine agreement protocol that matches
our lower bound for the authenticated case.
That is, we prove the following:

\begin{theorem}
\label{theorem:sync}

There is a deterministic synchronous authenticated Byzantine agreement
protocol with $O(ft+t)$ communication complexity.

\end{theorem}

\noindent We believe these results are interesting on their own
since they are the first to consider the actual number of failures, which
was previously considered in the problem of early
decision/stopping~\cite{dolev1990early, keidar2003simple}, for
communication complexity analysis of the Byzantine agreement problem.\\

\noindent\textbf{Optimal optimistic Byzantine agreement.}
Our final contribution is an optimistic Byzantine
agreement protocol that tolerates up to $t < n/3$ failures
and has asymptotically optimal communication cost under all network
conditions and failure scenarios.
That is, we prove the following:
\begin{theorem}
\label{theorem:opt}

There is an authenticated Byzantine agreement protocol with $O(ft+t)$
communication complexity in synchronous runs and expected $O(t^2)$ communication
complexity in all other runs.

\end{theorem}
To achieve the result, we combine our optimal adaptive synchronous
protocol with an asynchronous fallback, for which we use a variant of
VABA~\cite{VABA}. As we shortly explain, the combination is not trivial
since we need to preserve safety even if parties decide in different parts of the
protocol, and implement an efficient mechanism to prevent honest parties
from moving to the fallback in synchronous runs.

\subsection{Technical overview}

The combination of our synchronous part with the asynchronous fallback
introduces two main challenges. The first challenge is to design a mechanism that (1) makes sure parties do not move to the fallback
unless necessary for termination, and (2) has 
$O(ft+t)$ communication complexity in synchronous runs.
The difficulty here is twofold: first, parties cannot always
distinguish between synchronous and asynchronous runs.
Second, they cannot distinguish between honest parties that complain
that they did not decide (due to asynchrony) in the first part and
Byzantine parties that complain because they wish to increase the
communication cost by moving to the asynchronous fallback.
To deal with this challenge, we implement a \emph{Help\&tryHalting}
procedure.
In a nutshell, parties try to avoid the fallback part by
helping complaining parties learn the decision value and move to the fallback
only when the number of complaints indicates that the run is not
synchronous.
This way, each Byzantine party in a synchronous run cannot increase
the communication cost by more than $O(n) = O(t)$, where $n$ is the total number of parties.

The second challenge in the optimistic protocol is to combine
both parts in a way that guarantees safety.
That is, since some parties may decide in the synchronous part and
others in the asynchronous fallback, we need to make sure they
decide on the same value.
To this end, we use the \emph{leader-based view (LBV)} abstraction,
defined in~\cite{ace}, as a building block for both parts.
The LBV abstraction captures a single view in a view-by-view agreement protocol such that one of its important properties is that a
sequential composition of them preserves safety.
For optimal communication cost, we adopt techniques
from~\cite{hotstuff} and~\cite{VABA} to implement the LBV abstraction
with an asymptotically linear cost $(O(n))$.

Our synchronous protocol operates up to $n$ sequentially composed
pre-defined linear LBV instances, each with a different leader. 
To achieve an optimal (adaptive to the number
of actual failures) cost, leaders invoke their LBVs only if they have
not yet decided.
In contrast to eventually synchronous protocols, the synchronous
part is designed to provide termination only in synchronous runs.
Therefore, parties do not need to be synchronized before views, but
rather move from one LBV to the next at pre-defined times.
As for the asynchronous fallback, we use the linear LBV building
block to reconstruct the VABA~\cite{VABA} protocol in a way that
forms a sequential composition of LBVs, which in turn allows a
convenient sequential composition with the synchronous part.

\subsection{Related work}

The idea of combining several agreement protocols is not new.
The notion of speculative
linearizability~\cite{guerraoui2012speculative} allows parties to
independently switch from one protocol to another, without requiring
them to reach agreement to determine the change of a protocol.
Aguilera and Toueg~\cite{aguilera1996randomization} presented an hybrid
approach to solve asynchronous crash-fault consensus by combining
randomization and unreliable failure detection.
Guerraoui et al~\cite{guerraoui2010next} defined an abstraction that
captures byzantine agreement protocols and presented a framework to
compose several such instances.
 
Some previous work on Byzantine agreement consider a fallback in the
context of the number rounds required for
termination~\cite{brasileiro2001consensus, martin2006fast,
song2008bosco}.
That is, in well-behaved runs parties decide in a single communication
round, wheres in all other runs they fallback to a mode that requires
more rounds to reach an agreement.
We, in contrast, are interested in communication complexity. 
To the best of our knowledge, our protocol is the first protocol that
adapts its communication complexity based on the actual number of
failures.

The combination of synchronous and asynchronous runs in the context of
Byzantine agreement was previously studied by Blum et
al.~\cite{blum2019synchronous}.
Their result is complementary to ours since they deal with optimal
resilience rather than optimal communication.
They showed lower and upper bounds on the number of failures that both
(synchronous and asynchronous) parts can tolerate.
For the lower bound, they showed that $t_a + 2t_s <n$, where $t_a$ and
$t_s$ is the threshold failure in asynchronous and synchronous runs,
respectively.
In our protocol $t_a = t_s < n/3$, which means that the protocol is 
optimal in the sense that neither $t_a$ or $t_s$ can be increased
without decreasing the other.
For the upper bound, they present a matching algorithm for any $t_a$ and
$t_s$ that satisfy the weak validity condition. 
Our protocol, in contrast, satisfy the more practical external validity
condition (see more details in the next section) with an optimal
communication cost.

As for asynchronous Byzantine agreement, the lower bound in~\cite{VABA}
shows that there is no protocol with optimal resilience and $o(n^2)$
communication complexity.
Two recent works by Cohen et al.~\cite{cohen2020not} and Blum et
al~\cite{blum2020asynchronous}. circumvent this lower bound by trading
optimal resilience.
That is, their protocols tolerate $f < (1 - \epsilon)n/3$ Byzantine
faults.
We consider in this paper optimal resilience and thus our protocol
achieves optimal communication complexity in asynchronous runs.

The use of cryptographic tools (e.g. PKI and threshold signatures
schemes) is very common in distributed computing to reduce round and
communication complexity.
To be able to focus on the distributed aspect of the problem, many
previous algorithms assume ideal cryptographic tools to avoid the
analysis of the small error probability induced by the security
parameter.
This includes the pioneer protocols for Byzantine
broadcast~\cite{dolev1983authenticated, dolev1985bounds} and binary
asynchronous Byzantine agreement~\cite{bracha1985asynchronous},
recent works on synchronous Byzantine
agreement~\cite{momose2020optimal, nayak2020improved}, and most of the
exciting practical algorithms~\cite{zyzzyva, pbft} including the
state-of-the-art communication efficient
ones~\cite{pbft,hotstuff,gueta2019sbft, tendermint}).
We follow this approach and assume ideal threshold signatures schemes
for better readability.

 \section{Model}
\label{sec:model}

Following practical solutions~\cite{pbft,
gueta2019sbft, hotstuff, zyzzyva, HoneyBadger}, we
consider a Byzantine message passing peer to peer model with a set
$\Pi$ of $n$ parties and a computationally bounded adversary that
corrupts up to $t < n/3$ of them, $O(t) = O(n)$.
Parties corrupted by the adversary are called \emph{Byzantine} and
may arbitrarily deviate from the protocol. 
Other parties are \emph{honest}.
To strengthen the result we consider an adaptive adversary for the
upper bound and static adversary for the lower bound.
The difference is that a \emph{static} adversary must decide what
parties to corrupt at the beginning of every execution, whereas an
\emph{adaptive} adversary can choose during the executions.

\textbf{Communication and runs.}
The communication links are reliable but controlled by the
adversary, i.e., all messages sent among honest parties are
eventually delivered, but the adversary controls the delivery time.
We assume a known to all parameter $\Delta$ and say that a run of a
protocol is \emph{eventually synchronous} if there is a \emph{global
stabilization time (GST)} after which all message sent among honest
parties are delivered within $\Delta$ time.
A run is \emph{synchronous} if GST occurs at time 0, and
\emph{asynchronous} if GST never occurs.\\

\textbf{The Agreement problem.}
Each party get an input value from the adversary from some domain
$\mathbb{V}$ and the Agreement problem exposes an API to \emph{propose} a
value and to output a \emph{decision}.
We are interested in protocols that never compromise safety and thus
require the following property to be satisfied in all runs:

\begin{itemize}
  
 \item Agreement: All honest parties that decide, decide on the same
 value.
  
\end{itemize}

Due to the FLP result~\cite{FLP85}, no deterministic agreement protocol can provide safety and liveness properties in all
asynchronous runs.
Therefore, in this paper, we consider protocols that guarantee
(deterministic) termination in all synchronous and eventually
synchronous runs, and provides a probabilistic termination in
asynchronous ones:

\begin{itemize}
  
 \item Termination: All honest parties eventually decide.
  
 \item Probabilistic-Termination: All honest parties decide with
 probability 1.
  
\end{itemize}

As for validity, honest parties must decide only on values from
some domain $\mathbb{V}$.
For the lower bounds, to strengthen them as
much as possible, we consider the binary case, which is the weakest
possible definition:

\begin{itemize}
  
 \item Binary validity: The domain of valid values $\mathbb{V} =
 \{0,1\}$, and if all honest parties propose the same value $v \in
 \mathbb{V}$, than no honest party decides on a value other than $v$.

\end{itemize}

For the upper bounds, we are interested in practical multi-valued
protocols. 
In contrast to binary validity, in a multi-valued Byzantine agreement
we need also to define what is a valid decision in the case that not all parties a priori agree (i.e., propose different values).
One option is Weak Validity~\cite{PSL80,
blum2019synchronous}, which allows parties to agree on a pre-defined
$\bot$ in that case.
This definition is well defined and makes sense for some use cases.
When Pease et al.~\cite{PSL80} originally defined it, they had in mind
a spaceship cockpit with 4 sensors that try to agree even if one is
broken (measures a wrong value).
However, as Cachin et al, explain in their paper~\cite{CachinSecure}
and book~\cite{malkhi2019concurrency}, this definition is useless for
SMR (and Blockchains) since if parties do not a priori agree, then they
can keep agreeing on $\bot$ forever leaving the SMR with no "real"
progress.

To solve the limitation of being able to agree on $\bot$, we
consider the external validity property that was first defined by
Cachin et al.~\cite{CachinSecure}, which is implicitly
or explicitly considered in most practical Byzantine agreement
solutions we are aware of~\cite{VABA, pbft, hotstuff, gueta2019sbft,
zyzzyva}.
Intuitively, with external validity, parties are allowed to decide on
a value proposed by any party (honest and Byzantine) as long as it is
valid by some external predicate (e.g., all transaction are valid in
the block).
To capture the above, we give a formal definition below.

\begin{itemize}
  
 \item External validity: The domain of valid values $\mathbb{V}$
 is unknown to honest parties. At the beginning of every run,
 each honest party gets a value $v$ with a proof $\sigma$
 that $v \in \mathbb{V}$ such that all other honest parties can
 verify.
  
\end{itemize}

Note that our definition rules out trivial solutions such as
simply deciding on some pre-defined externally valid value because
the parties do not know what is externally valid unless they see a
proof.

We define an \emph{optimistic Agreement protocol} to be a protocol
that guarantees Agreement and External validity in all
runs, Termination in all synchronous and eventually synchronous runs, and Probabilistic-Termination in asynchronous runs.\\
 
 \textbf{Cryptographic assumptions.}
We assume a computationally bounded adversary and a trusted dealer
that equips parties with cryptographic schemes.
Following a common standard in distributed computing and for simplicity
of presentation (avoid the analysis of security
parameters and negligible error probabilities), we assume that the
following cryptographic tools are perfect:

\begin{itemize}

 \item \textbf{Authenticated link.} If an honest party $p_i$
 delivers a messages $m$ from an honest party $p_j$, then $p_j$ 
 previously sent $m$ to $p_i$.
 \item \textbf{Threshold signatures scheme.} We assume that each
 party $p_i$ has a private function $\emph{share-sign}_i$, and we
 assume 3 public functions: \emph{share-validate},
 \emph{threshold-sign}, and \emph{threshold-validate}. 
 Informally, given ``enough'' valid shares, the function
 \emph{threshold-sign} returns a valid threshold signature.
 For our algorithm, we sometimes require ``enough'' to be $t+1$ and
 sometimes $n-t$.
 A formal definition is given in Appendix~\ref{app:TS}.
  
\end{itemize}

\noindent We note that perfect cryptographic schemes do not exist in
practice.
However, since in real-world systems they often treated as such, we
believe that they capture just enough in order to be able to focus on
the distributed aspect of the problem.
Moreover, all the lower bounds in this paper hold even if protocols can
use perfect cryptographic schemes.
Thus, the upper bounds are tight in this aspect.\\

\textbf{Communication complexity.}
We denote by $f$ the actual number of corrupted parties in a given
run and we are interested in optimistic protocols that
utilize $f$ and the network condition to reduce communication cost.
Similarly to~\cite{VABA}, we say that a \emph{word} contains a
constant number of signatures and values, and each message contains
at least $1$ word.
The \emph{communication cost of a run $r$} is the number of words sent
in messages by honest parties in $r$.
For every $0 \leq f \leq t$, let $R^s_f$ and $R^{es}_f$ be the sets of
all synchronous and eventually synchronous runs with $f$ corrupted
parties, respectively.
The \emph{synchronous and eventually synchronous communication cost
with $f$ failures} is the maximal communication cost of runs in
$R^s_f$ and $R^{es}_f$, respectively.
We say that the \emph{synchronous communication cost of a
protocol A} is $G(f,t)$ if for every $0 \leq f \leq t$, its
synchronous communication cost with $f$ failures is $G(f,t)$.
The \emph{asynchronous communication cost of a protocol A} is the
expected communication cost of an asynchronous run of $A$. 

 \section{Lower Bounds}
\label{sec:lower}

In this section, we present two lower bounds on the communication complexity of deterministic Byzantine agreement protocols in
synchronous and eventually synchronous runs.

\subsection{Eventually synchronous runs}
The following theorem exemplifies the inherent vulnerability of the
eventually synchronous approach.

\begin{reptheorem}{theorem:ES}[restated]
There is no eventually synchronous deterministic Byzantine agreement
protocol that can tolerate a single failure and guarantee bounded
communication cost even in failure-free runs.
\end{reptheorem}

\begin{proof}

Assume by a way of contradiction that there are such algorithms.
Let $A$ be such an algorithm with the lowest eventually synchronous
communication cost with $0$ failures, and denote its communication
cost by $N$.
Clearly, $N \geq 1$.
Let $R_N \subset R^{es}_0$ be the set of all
failure-free eventually synchronous runs of $A$ that have
communication cost of $N$.
For every run $r \in R_N$ let $m_r$ be the last message that is
delivered in $r$, let $t_r$ be the time at which it is delivered, and let $p_r$
be the party that sends $m_r$.
Now for every $r \in R_N$ consider a run $r'$ that is identical to
$r$ up to time $t_r$ except $p_r$ is Byzantine that acts exactly as
in $r$ but does not send $m_r$.
Denote by $R_{N-1}$ the set of all such runs and consider two cases:

\begin{itemize}
  
 \item There is a run $r' \in R_{N-1}$ in which some message $m$ by an
 honest party $p$ is sent at some time $t_{r'} > t_r$.
 Now consider a failure-free run $r''$ that is identical to run $r$
 except the delivery of $m_r$ is delayed to $t_{r'}+1$.
 The runs $r''$ and $r'$ are indistinguishable to all parties that
 are honest in $r'$ and thus $p$ sends $m$ at time 
 time $t_{r'} > t_r$ in $r''$ as well.
 Therefore, the communication cost of $r''$ is at least $N+1$. 
 A contradiction to the communication cost of $A$.
  
 \item Otherwise, we can construct an algorithm $A'$ with a better
 eventually synchronous communication cost with $0$ failures than
 $A$ in the following way: $A'$ operates identically to $A$ in all
 runs not in $R_N$ and for every run $r \in R_N$ $A'$ operates as
 $A$ except $p_r$ does not send $m_r$.
 A contradiction to the definition of $A$.  
  
\end{itemize}

\end{proof}

\subsection{Synchronous runs} We next prove a lower bound that applies
even to synchronous Byzantine agreement algorithms and is adaptive to
the number of actual failures $f$.
The proof is a generalization of the proof
in~\cite{dolev1985bounds}, which has been proved for the  
Byzantine broadcast problem and considered the worst-case scenario
($f=t$).
It is important to note that the proof captures deterministic
authenticated algorithms even if they are equipped with perfect
cryptographic tools.

We start with a simple claim:

\begin{claim}
\label{claim:atLeastT}

The synchronous communication cost with 0 failures of any Byzantine
agreement algorithm is at least~$t$.

\end{claim}

\begin{proof}

Assume by a way of contradiction such algorithm $A$ that sends less
than $t$ messages in runs with $0$ failures and consider a run $r \in
R^s_0$ of $A$ in which all parties propose 1.
By the Termination and Binary validity properties, all
parties decide 1 in $r$.
By the contradicting assumption and since all parties in $r$ are
honest, there are at least $2t+1$ honest parties that get no messages
in $r$.
Now consider anther run $r' \in R^s_0$ of $A$ in which all parties
propose 0, and again, there are at least $2t+1$ honest parties that get
no messages in $r'$.
Thus, there is at least one honest party $p$ that gets no messages in
both runs and thus cannot distinguish between $r$ and $r'$.
Therefore, $p$ decides 1 in run $r'$ as well.
A contradiction to the Binary validity property. 

\end{proof}

The following Lemma shows that if honest parties send 
$o(ft)$ messages, then Byzantine parties can prevent honest parties
from getting any of them.

\begin{lemma}
\label{lem:synchNoMessages}

Assume that there is a Byzantine agreement algorithm $A$,
which synchronous communication cost with $f$ failures is
$o(ft)$ for some $1 \leq f \leq \floor*{t/2}$.
Then, for every set $S \subset \Pi$ of $f$ parties and every set
of values proposed by honest parties, there is a synchronous run $r'$
s.t.\ some honest party $p \in S$ does not get any messages in $r'$.

\end{lemma}

\begin{proof}

Let $r \in R^s_f$ be a run in which all parties in $S$ are
Byzantine that (1) do not send messages among themselves, and (2)
ignore all messages they receive and act like honest parties that get
no messages.
By the assumption, there is a party $p \in S$ that receives
less than $t/2$ messages from honest parties in $r$. 
Denote the set of (honest) parties outside $S$ that send messages to
$p$ in $r$ by $P\subset \Pi\setminus S$ and consider the following run
$r'$:

\begin{itemize}
  
  \item Parties in $S \setminus \{p\}$ are Byzantine that act like
  in $r$.
    
    \item Parties in $P$ are Byzantine. They do not send messages to
    $p$, but other than that act as honest parties.
    
    \item All other parties, including $p$, are honest.
  
\end{itemize}

\noindent First, note that the number of Byzantine parties in $r'$ is
$|S|-1 + |P| \leq f-1 + t/2 \leq t$.
Also, since $p$ acts in $r$ as an honest party that does not receive
messages, and all Byzantine parties in $r'$ act towards honest
parties in $r'$ ($\Pi \setminus (S \cup P)$) in exactly the same way
as they do in $r$, then honest parties in $r'$ cannot distinguish
between $r$ and $r'$. 
Thus, since they do not send messages to $p$ in $r$ they do not send
in $r'$ as well.
Therefore, $p$ does not get any message in $r'$.

\end{proof}

The next Lemma is proven by showing that honest parties that do not get
messages cannot safely decide.
Not that the case of $f > t/2$ is not required to conclude
Theorem~\ref{theorem:S} since in this case $o(ft) = o(t^2)$.

\begin{lemma}
\label{lem:sunchLemma}

For any $1 \leq f \leq \floor*{t/2}$, there is no optimistic
Byzantine agreement algorithm which synchronous communication cost
with $f$ failures is $o(ft)$.

\end{lemma}

\begin{proof}

Assume by a way of contradiction such protocol $A$ which synchronous
communication cost with $f$ failures is $o(ft)$ for some $1
\leq f \leq \floor*{t/2}$. 
Pick a set of $S_1 \subset \Pi$ of $f$ parties and let $V$ be the set
of values that honest parties propose.
By Lemma~\ref{lem:synchNoMessages}, there is a run $r_1$ of $A$ in which
honest parties propose values from $V$ s.t.\ some honest party $p_1
\in S$ does not get any messages.
Now let $S_2 = \{p\} \cup S_1\setminus\{p_1\}$ s.t.\ $p
\in \Pi \setminus S_1$.
By Lemma~\ref{lem:synchNoMessages} again, there is a run $r_2$ of $A$
in which honest parties propose values from $V$ s.t.\ some honest party
$p_2 \neq p_1$ does not get any messages.
Since $f \leq \floor*{t/2}$, we can repeat the above $2t + 1$ times by
each time replacing the honest party in $S_i$ that get no messages with
a party not in $S_i \cup \{p_1,p_2,\ldots,p_i\}$. 
Thus, we get that for every possible set of inputs $V$ (values proposed
by honest parties) there is a set $T$ of $2t+1$ parties s.t.\ for
every party $p \in T$ there is a run of $A$ in which honest parties
propose values from $V$, $p$ is honest, and $p$ does not get any
messages.
In particular, there exist such set $T_0$ for the case in which all
honest party input $0$ and a set $T_1$ for the case in which all
honest parties input $1$.
Since $|T_0| = |T_1| = 2t+1$, there is a party $p \in T_1 \cap T_2$.
Therefore, by the Termination and Binary validity properties, there
is a run $r$ in which $p$ does not get any messages and decides $0$
and a run $r'$ in which $p$ does not any messages and decides $1$.
However, since $r$ and $r'$ are indistinguishable to $p$ we get a
contradiction.

\end{proof}

\noindent The following Theorem follows directly from
Lemma~\ref{lem:sunchLemma} and Claim~\ref{claim:atLeastT}.

\begin{reptheorem}{theorem:S}[restated]
Any synchronous deterministic Byzantine agreement protocol has a
communication cost of $\Omega(ft + t)$.
\end{reptheorem}

 \section{Asymptotically optimal optimistic Byzantine Agreement}
\label{sec:algorithms}

Our optimistic Byzantine agreement protocol safely combines synchronous and asynchronous protocols.
Our synchronous protocol, which is interesting on its own, matches the
lower bound proven in Theorem~\ref{theorem:S}.
That is, its communication complexity is $O(ft+t)$.
The asynchronous protocol we use has a worst-case optimal quadratic
communication complexity.
For ease of exposition, we construct our protocol in steps.
First, in Section~\ref{sub:algGeneral}, we present the local state
each party maintains, define the \emph{leader-based view
(LBV)}~\cite{ace} building block, which is used by both protocols, and
present an implementation with $O(n)$ communication complexity.
Then, in Section~\ref{sub:algSynch}, we describe our synchronous
protocol, and in Section~\ref{sub:algAsynch} we use the LBV building
block to reconstruct VABA~\cite{VABA} - an asynchronous Byzantine
agreement protocol with expected $O(n^2)$ communication cost and $O(1)$
running time.
Finally, in section~\ref{sub:algOptimistic}, we safely
combine both protocols to prove the following:

\begin{reptheorem}{theorem:opt}[restated]

There is an authenticated Byzantine agreement protocol with $O(ft+t)$
communication complexity in synchronous runs and expected $O(t^2)$ communication
complexity in all other runs. 

\end{reptheorem}

\noindent A formal correctness proof and communication analysis of the
protocol appear in Appendix~\ref{app:upper}.

\subsection{General structure}
\label{sub:algGeneral}

The protocol uses many instances of the LBV building
block, each of which is parametrized with
a sequence number and a leader. 
We denote an LBV instance that is parametrized with sequence number $sq$
and a leader $p_l$ as $\textit{LBV}(sq,p_l)$. Each party in the protocol
maintains a local state, which is used by all LBVs and is updated
according to their returned values.
Section~\ref{subsub:state} presents the local state and 
Section~\ref{subsub:LBV} describes a linear communication LBV
implementation.
Section~\ref{sub:composition} discusses the properties guaranteed
by a sequential composition of several LBV instances.

\subsubsection{Local state}
\label{subsub:state}

The local state each party maintains is presented in
Algorithm~\ref{alg:state}.
For every possible sequence number $sq$, $\emph{LEADER}[sq]$ stores
the party that is chosen (a priori or in retrospect) to be the leader associated with $sq$.
The \emph{COMMIT} variable is a tuple that consists of a value $val$,
a sequence number $sq$ s.t.\ $val$ was
committed in LBV(sq,\emph{LEADERS}[sq]), 
and a threshold signature that is used as a proof of it.
The \emph{VALUE} variable contains a safe value to propose and the
\emph{KEY} variable is used as proof that \emph{VALUE} is indeed
safe.
\emph{KEY} contains a sequence number $sq$ and a threshold signature
that proves that no value other than \emph{VALUE} could be committed
in LBV(sq,\emph{LEADERS}[sq]).
The \emph{LOCK} variable stores a sequence number $sq$, which is used
to determine what keys are up-to-date and what are obsolete -- a key
is up-to-date if it contains a sequence number that is greater than
or equal to \emph{LOCK}.

\begin{algorithm}
\caption{Local state initialization.}
 \begin{algorithmic}[1]
\small

\Statex $\emph{LOCK} \in \mathbb{N} \cup \{\bot\}$, initially $\bot$
\Statex $\emph{KEY} \in (\mathbb{N} \times
\{0,1\}^*) \cup\{\bot\}$ with selectors \emph{sq} and \emph{proof},
initially $\bot$ 
\Statex $\emph{VALUE} \in \mathbb{V}\cup
\{\bot\}$, initially $\bot$
\Statex $\emph{COMMIT} \in  (\mathbb{V} \times
\mathbb{N}\times \{0,1\}^*)\cup\{\bot\}$ with selectors \emph{val},
\emph{sq} and \emph{proof}, initially $\bot$  
\Statex \textbf{for every} $sq \in \mathbb{N}$, $\emph{LEADER}[sq]
\in \Pi \cup \{\bot\}$, initially $\bot$

\end{algorithmic}
\label{alg:state}
\end{algorithm}

\subsubsection{Linear leader-based view}
\label{subsub:LBV}

Detailed pseudocode of the linear
implementation of the LBV building block is given in
Algorithms~\ref{alg:LBV-API} and~\ref{alg:LBV-messages}.
An illustration appears in figure~\ref{fig:LBV}.
The LBV building block supports an API to \emph{start the view}
and \emph{wedge the view}.
Upon a \lv{\lr{sq, p_l}} invocation, the invoking party starts
processing messages associated with LBV(sq,$p_l$).
When the leader $p_l$ invokes \lv{\lr{sq, p_l}} it initiates $3$ 
steps of leader-to-all and all-to-leader communication, named
\emph{PreKeyStep}, \emph{KeyStep}, and \emph{LockStep}.
In each step, the leader sends its \emph{VALUE}
together with a threshold signature that proves the safety of the
value for the current step and then waits to collect $n-t$ valid
replies.
A party that gets a message from the leader, validates that the
received value and proof are valid for the current step, then
produces its signature share on a message that contains the value and
the step's name, and sends the share back to the leader.
When the leader gets $n-t$ valid shares, it combines them into a
threshold signature and continues to the next step.
After successfully generating the threshold signature at the end of
the third step (\emph{LockStep}), the leader has a commit certificate
which he sends together with its \emph{VALUE} to all parties.

\begin{algorithm}[H]
\caption{A linear implementation of LBV(sq,\emph{leader}):
API for a party $p_i$.}
\begin{algorithmic}[1]
\footnotesize

\Statex \textbf{Local variables initialization:}
\Statex \hspace*{0.4cm} $S_{\emph{key}} = S_{\emph{lock}} =
S_{\emph{commit}} = \{\}$
\Statex \hspace*{0.4cm} $\textbf{keyProof}, \textbf{lockProof},
\textbf{commitProof} \in (\mathbb{V} \times \{0,1\}^*) \cup\{\bot\}$
\hspace*{3cm} with selectors \emph{val} and \emph{proof}, initially
$\bot$ \Statex \hspace*{0.4cm} $\emph{active} \gets \emph{true}$
; $\emph{done} \gets \emph{false}$




\Upon{wedgeView($\emph{sq}, \emph{leader}$)
invocation}

	\State $\emph{active} \gets \emph{false}$
	\State return $\lr{\textbf{keyProof}, \textbf{lockProof},
	\textbf{commitProof}}$

\EndUpon

\Statex

\Upon{startView($\emph{sq}, \emph{leader}$) invocation}

\State start processing received messages associated with \emph{sq}
and
\emph{leader}
\If{\emph{leader} = $p_i$}
	
	\State \hspace*{1cm} \textcolor{gray}{**//first step//**}
	\State send ``$\textsc{preKeyStep}, sq, leader,
	\emph{VALUE}, \emph{KEY}$'' to all parties
	\State \textbf{wait} until $|S_{\emph{key}}| = n-t$
	\State $\nu_k \gets \emph{threshold-sign}(S_{\emph{key}})$
	
	\Statex
	\State \hspace*{1cm} \textcolor{gray}{**//second step//**}
	\State send ``$\textsc{KeyStep}, sq, leader,
	\emph{VALUE}, \nu_k$'' to all parties
	\State \textbf{wait} until $|S_{\emph{lock}}| = n-t$
	\State $\nu_l \gets \emph{threshold-sign}(S_{\emph{lock}})$
	
	\State \hspace*{1cm} \textcolor{gray}{**//third step//**}
	\State send ``$\textsc{lockStep}, sq, leader,
	\emph{VALUE}, \nu_l$'' to all parties
	\State \textbf{wait} until $|S_{\emph{commit}}| = n-t$
	\State $\nu_c \gets \emph{threshold-sign}(S_{\emph{commit}})$
	
	\State \hspace*{1cm} \textcolor{gray}{**//broadcast the commit//**}
	\State send ``$\textsc{commit}, sq, leader,
	\emph{VALUE}, \nu_c$'' to all parties 

\EndIf

\State \textbf{wait} for $\emph{done} = \emph{true}$
\State return $\lr{\textbf{keyProof}, \textbf{lockProof},
	\textbf{commitProof}}$ 
 
\EndUpon

\Statex


\end{algorithmic}
\label{alg:LBV-API}
\end{algorithm}

In addition to validating and share-signing messages, parties also
store the values and proofs they receive.
The \textbf{keyProof} and \textbf{lockProof} variables store tuples
consisting of the values and the threshold signatures received from the
leader in the \emph{KeyStep}, and \emph{LockStep} steps, respectively.
The \textbf{commitProof} variable stores the received value and the
commit certificate.
When a party receives a valid commit certificate from the leader
it returns.

As for the validation of the leader's messages, parties distinguish
the \emph{PreKeyStep} message from the rest.
For \emph{KeyStep}, \emph{LockStep} and commit certificate messages,
parties simply check that the attached proof is a valid threshold
signature on the leader's value and the previous step name. 
The \emph{PreKeyStep} message, however, is used by the Agreement
protocols to safely compose many LBV instances.
We describe this mechanism in more details below, but to develop some
intuition let us first present the properties guaranteed by a single
LBV instance:

\begin{itemize}
  
 \item Commit causality: If a party gets a valid commit certificate,
 then at least $t+1$ honest parties previously got a valid
 \textbf{lockProof}.
  
 \item Lock causality: If a party gets a valid \textbf{lockProof},
 then at least $t+1$ honest parties previously got a valid
 \textbf{keyProof}.
  
 \item Safety: All valid \textbf{keyProof}, \textbf{lockProof}, and
 commit certificates obtained in the same LBV have the same value.
  
\end{itemize}

\begin{algorithm}
\caption{A linear implementation of LBV(sq,\emph{leader}):
$p_i$'s message handlers.}
\begin{algorithmic}[1]
 \footnotesize

\Receiving{``$\textsc{preKeyStep}, sq, \emph{leader}, \emph{value},
\emph{key}$'' from \emph{leader}}
	\If{\emph{active} $\wedge$ ($\emph{LOCK} = \bot $ $\vee$
	$\emph{key.sq} \geq \emph{LOCK}$)}
	\Comment{key is up-to-date}
	
		\If{$\emph{key} = \bot$ $\vee$
		$\emph{threshold-validate}(\lr{\textsc{preKeyStep}, \emph{key.sq},
		$\\ \hspace*{2cm} $
		\emph{LEADERS}[\emph{key.sq}],\emph{value}}, \emph{key.proof})$)}
			
			\State $\rho_k \gets \emph{share-sign}_i(\lr{\textsc{preKeyStep},
			sq, \emph{leader}, \emph{value}})$
			 
			\State send ``$\textsc{keyShare}, sq, leader, \rho_k$'' to
			\emph{leader}
		
		\EndIf
	
	\EndIf 

\EndReceiving

\Statex

\Receiving{``$\textsc{keyShare}, sq, leader, \rho_k$'' from 
$p_j$ for the first time}

	\If{$leader = p_i$}
	
		\If{\emph{share-validate}($\lr{\textsc{preKeyStep}, sq, \emph{leader},
	\emph{VALUE}}, p_j, \rho_k)$}
	
			\State $S_{\emph{key}} \gets S_{\emph{key}} \cup \rho_k$	
	
		\EndIf
	\EndIf

\EndReceiving

\Statex

\Receiving{``$\textsc{KeyStep}, sq, \emph{leader}, \emph{value},
\nu_k$'' from \emph{leader}}
	\If{\emph{active}}
		\If{$\emph{threshold-validate}(\lr{\textsc{preKeyStep}, sq,
	\emph{leader}, \emph{value}}, \nu_k)$}
			\State $\textbf{keyProof} \gets \lr{value, \nu_k}$
			\State $\rho_l \gets \emph{share-sign}_i(\lr{\textsc{KeyStep},
			sq, \emph{leader}, \emph{value}})$
			 
			\State send ``$\textsc{lockShare}, sq, leader, \rho_l$'' to
			\emph{leader}
		\EndIf
	
	\EndIf 

\EndReceiving

\Statex

\Receiving{``$\textsc{lockShare}, sq, leader, \rho_l$'' from
$p_j$ for the first time}

	\If{$leader = p_i$}
	
		\If{\emph{share-validate}($\lr{\textsc{KeyStep}, sq, \emph{leader},
 		\emph{VALUE}}, p_j, \rho_l)$}
		
			\State $S_{\emph{lock}} \gets S_{\emph{lock}} \cup \rho_l$	
		
		\EndIf
	
	\EndIf

\EndReceiving

\Statex

\Receiving{``$\textsc{lockStep}, sq, \emph{leader}, \emph{value},
\nu_l$'' from \emph{leader}}
	\If{\emph{active}}
	
		\If{$\emph{threshold-validate}(\lr{\textsc{KeyStep}, sq,
	\emph{leader}, \emph{value}}, \nu_l)$}
			\State $\textbf{lockProof} \gets \lr{value, \nu_l}$
			\State $\rho_c \gets \emph{share-sign}_i(\lr{\textsc{lockStep},
			sq, \emph{leader}, \emph{value}})$
			 
			\State send ``$\textsc{commitShare}, sq, leader, \rho_c$'' to
			\emph{leader}
		\EndIf
	
	\EndIf 

\EndReceiving

\Statex

\Receiving{``$\textsc{commitShare}, sq, leader, \rho_c$'' from
$p_j$ for the first time}

	\If{$leader = p_i$}
	
		\If{\emph{share-validate}($\lr{\textsc{lockStep}, sq, \emph{leader},
		\emph{VALUE}}, p_j, \rho_c)$}
	
			\State $S_{\emph{commit}} \gets S_{\emph{commit}} \cup \rho_c$	
		\EndIf
	\EndIf

\EndReceiving

\Statex

\Receiving{``$\textsc{commit}, sq, \emph{leader}, \emph{value},
\nu_c$'' from \emph{leader}}
	\If{\emph{active}}
	
		\If{$\emph{threshold-validate}(\lr{\textsc{lockStep}, sq,
		\emph{leader}, \emph{value}}, \nu_c)$}
				\State $\textbf{commitProof} \gets \lr{value, \nu_c}$
 				\State $\emph{done} \gets \emph{true}$
	\EndIf
	
	\EndIf 

\EndReceiving

\end{algorithmic}
\label{alg:LBV-messages}
\end{algorithm}

The validation of the \emph{PreKeyMessage} in \emph{PreKeyStep} makes
sure that the leader's value satisfies the safety properties of the
Byzantine agreement protocol that sequentially composes and
operates several LBVs.
The \emph{PreKeyMessage} contains the leader's \emph{VALUE} and
\emph{KEY}, where \emph{KEY} stores the last (non-empty)
\textbf{keyProof} returned by a previous LBV instance together with
the LBV's sequence number.
When a party gets a \emph{PreKeyMessage} it first validates, by
checking the key's sequence number $sq$, that the attached key was
obtained in an LBV instance that does not precede the one the party is locked on (the sequence number that is stored in the party's
\emph{LOCK} variable).
Then, the party checks that the threshold signature in the key (1) was
generated at the end of the \emph{PreKeyStep} step (it is a valid
\textbf{keyProof}) in LBV(sq,\emph{LEADER[sk]}); and (2) it is a
valid signature on a message that contains the leader's \emph{VALUE}.
Note that if the party is not locked ($\emph{LOCK} = \bot$) then a key is not required.

Upon a \wv{sq, p_l} invocation, the invoking party stops
participating in LBV(sq,$p_l$) and returns its current 
\textbf{keyProof}, \textbf{lockProof}, and \textbf{commitProof} values.
These values are used by both synchronous and asynchronous protocols,
which are built on top of LBV instances, to update the \emph{LOCK},
\emph{KEY}, \emph{VALUE}, and \emph{COMMIT} variables in parties' local
states.
Stopping participating in LBV(sq,$p_l$) upon a \wv{sq, p_l} invocation
guarantees that the the LBVs' causality guarantees are propagated the
\emph{KEY}, \emph{LOCK}, and \emph{COMMIT} variables in parties
local states.\\

\textbf{Communication complexity.}
Note that the number of messages sent among honest parties in an LBV
instance is $O(n) = O(t)$.
In addition, since signatures are not accumulated -- leaders use
threshold signatures -- each
message contains a constant number of words, and thus the total
communication cost of an LBV instance is $O(t)$ words.

\begin{figure*}[th]
  \begin{center}
    \includegraphics[width=0.75\textwidth]{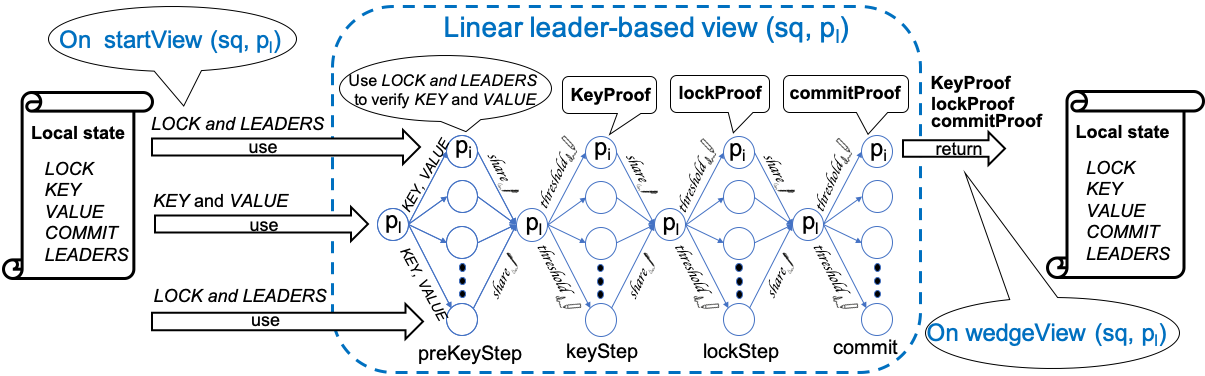}
  \end{center}
  \caption{A linear communication LBV illustration. The local state is
  used by and updated after each instance. The \textbf{keyProof},
  \textbf{lockProof}, and \textbf{commitProof} are returned when a
  commit message is received from the leader or \emph{wedgeView} is
  invoked.}
  \label{fig:LBV}
\end{figure*} 

\subsubsection{Sequential composition of LBVs}
\label{sub:composition}

As mentioned above, our optimistic Byzantine agreement protocol is
built on top of the LBV building blocks.
The synchronous and the asynchronous parts of the protocol use different approaches, but they both sequentially compose
LBVs - the synchronous part of the protocol determines the
composition in advance, whereas the asynchronous part chooses what
instances are part of the composition in retrospect.

In a nutshell, a sequential composition of LBVs operates as
follows: parties start an LBV instance by invoking
$\mathsf{startView}$ and at some later time (depends on the approach)
invoke $\mathsf{wedgeView}$ and update their local states with the
returned values.
Then, they exchange messages to propagate information (e.g.,
up-to-date keys or commit certificates), update their local states
again and start the next LBV (via $\mathsf{startView}$ invocation).
We claim that an agreement protocol that sequentially
composes LBV instances and maintains the local state in
Algorithm~\ref{alg:state} has the following properties:

\begin{itemize}
  
 \item Agreement: all commit certificates in all LBV instances have
 the same value.

 \item Conditional progress: for every LBV instance, if the leader is
 honest, all honest parties invoke \emph{startView}, and all messages
 among honest parties are delivered before some honest party invokes
 \emph{wedgeView}, then all honest parties get a commit certificate.
  
\end{itemize}

Intuitively, by the LBV's commit causality property,
if some party returns a valid commit certificate (\textbf{commitProof})
with a value $v$ in some LBV(sq,$p_i$), then at least $t+1$ honest
parties return a valid \textbf{lockProof} and thus lock on $sq$
($\emph{LOCK} \gets sq $).
Therefore, since the leader of the next LBV needs the cooperation of
$n-t$ parties to generate threshold signatures, its
\emph{PreKeyStep} message must include a valid \textbf{keyProof} that
was obtained in LBV(sq,$p_i$).
By the LBV's safety property, this \textbf{keyProof}
includes the value $v$ and thus $v$ is the only value the leader can
propose.
The agreement property follows by induction.

As for conditional progress, we have to make sure that honest leaders
are able to drive progress. 
Thus, we must ensure that all honest leaders have the most
up-to-date keys.
By the lock causality property, if some party gets a valid
\textbf{lockProof} in some LBV, then at least $t+1$ honest
parties get a valid \textbf{keyProof} in this LBV and thus are able
to unlock all honest parties in the next LBV.
Therefore, leaders can get the up-to-date key by 
querying a quorum of $n-t$ parties. 

From the above, any Byzantine agreement protocol that sequentially
composes LBVs satisfies Agreement.
The challenge, which we address in the rest of this section, is how
to sequentially compose LBVs in a way that satisfies
Termination with asymptotically optimal communication
complexity under all network conditions and failure scenarios.

\subsection{Adaptive to failures synchronous protocol}
\label{sub:algSynch}

\begin{algorithm}
\caption{Adaptive synchronous protocol: Procedure for a party~$p_i$.}


\begin{algorithmic}[1]
\footnotesize

\Upon{Synch-propose($v_i$)}

	\State $\emph{VALUE} \gets v_i$
	\State \emph{tryOptimistic()}

\EndUpon

\Statex 

\Procedure{tryOptimistic()}{}

\State $\mathsf{trySynchrony(1, p_1, 7\Delta)}$

\For{$j \gets 2 $ to $n$}

	\If{$i \neq j$}
	
		\State $\mathsf{trySynchrony(j, p_j, 9\Delta)}$
	
	\ElsIf{$\emph{COMMIT} = \bot $}
	
		\State send ``$\textsc{keyRequest}$'' to all parties
		\State \textbf{wait} for $2\Delta$ time 
		\State $\mathsf{trySynchrony(j, p_j, 7\Delta)}$

	\EndIf

\EndFor

\EndProcedure

\Statex

\algstore{myalg}
\end{algorithmic}

  \begin{algorithmic}[1]
 \footnotesize
\algrestore{myalg}

\Procedure{trySynchrony}{$sq, leader, T$}

	\State invoke \lv{\emph{sq},\emph{leader}}
	\Comment{non-blocking invocation}
	\State \textbf{wait} for $T$ time
	\State $\lr{\emph{keyProof}, \emph{lockProof}, \emph{commitProof}}
	\gets$ \wv{\emph{sq},\emph{leader}} 
	\State $\mathsf{updateState(\emph{sq},\emph{leader}, \emph{keyProof},
	\emph{lockProof}, \emph{commitProof})}$

\EndProcedure

\Statex

\Receiving{``$\textsc{keyRequest}$'' from party $p_k$ for the first
time}

	\State send ``$\textsc{keyReply}, \emph{KEY}, \emph{VALUE}$'' to
	party $p_k$

\EndReceiving

\Statex

\Receiving{``\textsc{keyReply}, \emph{key}, \emph{value}''}

	\State $\mathsf{check\&updateKey(key, value)}$

\EndReceiving


\end{algorithmic}
\label{alg:synch}
\end{algorithm}

In this section, we describe a synchronous Byzantine agreement
protocol with an asymptotically optimal adaptive communication cost that matches the lower bound in Theorem~\ref{theorem:S}.
Namely, we prove the following Theorem:

\begin{reptheorem}{theorem:sync}[restated]

There is a deterministic synchronous authenticated Byzantine agreement
protocol with $O(ft+t)$ communication complexity.

\end{reptheorem}

\begin{figure*}[th]
  \begin{center}
    \includegraphics[width=0.85\textwidth]{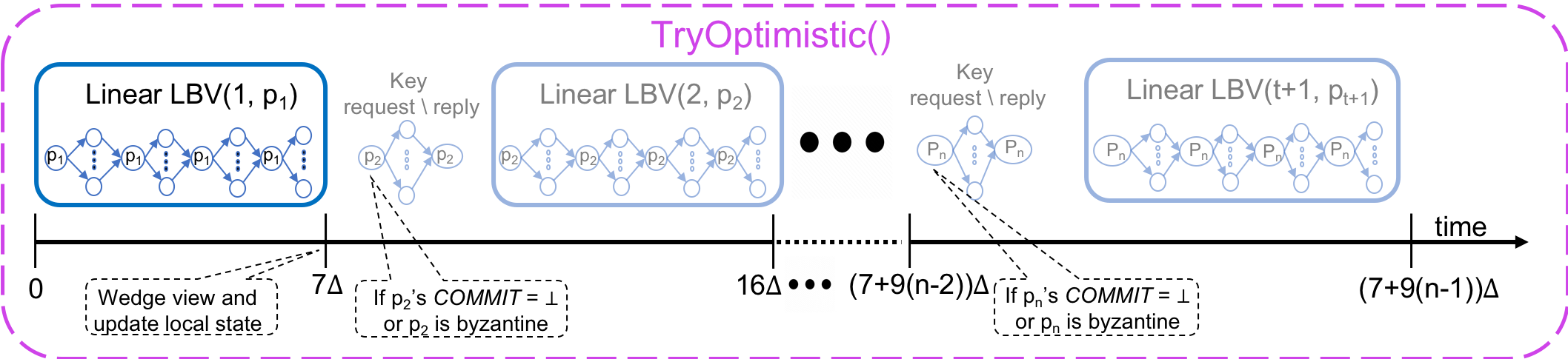}
  \end{center}
  \caption{Illustration of the adaptive synchronous protocol.
  Shaded LBVs are not executed if their leaders have previously
  decided.}
  \label{fig:synch}
\end{figure*} 

A detailed pseudocode is given in Algorithms~\ref{alg:synch}
and~\ref{alg:synchauxiliary}, and an illustration appears in
figure~\ref{fig:synch}.
The protocol sequentially composes $n$ pre-defined LBV
instances, each with a different leader, and parties decide $v$
whenever they get a commit certificate with $v$ in one of them.
To avoid the costly view-change mechanism that is usually unavoidable
in leader-based protocols, parties exploit synchrony to coordinate
their actions.
That is, all the
$\mathsf{startView}$ and $\mathsf{wedgeView}$ invocation times are
predefined, e.g., the first LBV starts at time 0 and is wedged at time
$7\Delta$ simultaneously by all honest parties.  
In addition, to make sure honest leaders can drive
progress, each leader (except the first) learns the up-to-date key,
before invoking $\mathsf{startView}$, by querying all parties and
waiting for a quorum of $n-t$ parties to reply.

\begin{algorithm}[h]
\caption{Auxiliary procedures to update local state.}
\begin{algorithmic}[1]
\footnotesize

\Procedure{updateState}{$sq, \emph{leader}, \emph{keyProof},
{\emph{lockProof}},
{\emph{commitProof}}$}

	\State $\emph{LEADERS}[sq] \gets \emph{leader}$
  	\If{$\emph{keyProof} \neq \bot$}
  	
  		\State $\emph{KEY} \gets \lr{sq, \emph{keyProof.proof}}$
  		\State $\emph{VALUE} \gets \emph{keyProof.val}$
  	
  	\EndIf
  	
  	\If{$\emph{lockProof} \neq \bot$}
  	
  		\State $\emph{LOCK} \gets sq$
  	
  	\EndIf
	
 	\If{$\emph{commitProof} \neq \bot$}

 		\State $\emph{COMMIT} \gets \lr{\emph{commitProof.val}, sq,
 		\emph{commitProof.proof}}$
 		\State \textbf{decide} \emph{COMMIT.val}
 	
 	\EndIf

\EndProcedure

\Statex
\Statex

\Procedure{check$\&$updateKey}{$key, value$}

	\If{$(\emph{KEY} = \bot \vee \emph{key.sq} > \emph{KEY.sq})$}
		
		\If{$\emph{threshold-validate}(\lr{\textsc{preKeyStep}, \emph{key.sq},
		$\\ \hspace*{2.3cm} $ \emph{LEADER}[\emph{key.sq}], \emph{value}}, \emph{key.proof})$}
		
			\State $\emph{KEY} \gets \emph{key}$
			\State $\emph{VALUE} \gets \emph{value}$
		
		\EndIf
		
	\EndIf

\EndProcedure

\Statex

\Procedure{check$\&$updateCommit}{\emph{commit}}

	\If{$\emph{COMMIT} = \bot$}

		\If{$\emph{threshold-validate}(\lr{\textsc{lockStep},
		\emph{commit.sq}, $\\ \hspace*{1cm} $
		\emph{LEADER}[\emph{commit.sq}],
		 \emph{commit.val}}, \emph{commit.proof})$}
			
			\State $\emph{COMMIT} \gets \emph{commit}$
			\State \textbf{decide} \emph{COMMIT.val} 
	
		\EndIf
	
	\EndIf
	
\EndProcedure

\end{algorithmic}
\label{alg:synchauxiliary}
\end{algorithm}

Composing $n$ LBV instances may lead in the
worst case to $O(t^2)$ communication complexity -- $O(t)$ for every LBV
instance.
Therefore, to achieve the optimal adaptive complexity, honest
leaders in our protocol participate (learn the up-to-date key and
invoke $\mathsf{startView}$) only in case they have not yet decided.
(Note that the communication cost of an LBV instance in which the leader does not invoke $\mathsf{startView}$ is 0 because other
parties only reply to the leader's messages.)
For example, if the leader of the second LBV instance is honest and
has committed a value in the first instance (its $\emph{COMMIT} \neq
\bot$ at time $7\Delta$), then no message is sent among honest
parties between time $7\Delta$ and time $16\Delta$.

\begin{algorithm}
\caption{Asynchronous fallback: Protocol for a party $p_i$.}
 
\begin{algorithmic}[1]
\footnotesize


\Upon{Asynch-propose($v_i$)}

	\State $\emph{VALUE} \gets v_i$
	\State \emph{fallback}($1$)

\EndUpon


\Statex

\Procedure{fallback}{$sq_{init}$}

	\CSTATE $\emph{RRleader} \gets 1$
	\State $sq \gets sq_{init} + 1$
	\While{\emph{true}}
	
		\State $\mathsf{wave}(sq)$
		\State \we{\emph{sq}}
		\State $\mathsf{help\&tryHalting}(sq)$
		\CSTATE $\mathsf{trySynchrony}(sq + 1, RRleader, 8\Delta)$
		\CSTATE \we{\emph{sq+1}}
		\CSTATE $\mathsf{help\&tryHalting(sq+1)}$
		\CSTATE $\emph{RRleader} \gets \emph{RRleader} + 1 \text{ mod }
		|\Pi|$ \State $sq \gets sq + 2$

	\EndWhile
	
\EndProcedure

\Statex 

\Upon{\lv{\emph{sq},p_j} returns}
 	\State send ``$\textsc{your-view-done}, sq$'' to party $p_j$
 
\EndUpon


 \Statex

\Procedure{wave}{$sq$}

	\ForAll{$p_j=p_1,\ldots,p_n$}
		\State invoke \lv{\emph{sq},p_j}
		\Comment{non-blocking invocation}
	\EndFor
	
	\State \barriersync{\emph{sq}}
 	\Comment{blocking}

	\State $\emph{leader}  \gets $ \elect{\emph{sq}}
	\State $\lr{\emph{keyProof}, \emph{lockProof}, \emph{commitProof}}
	 \gets$ \wv{sq, \emph{leader}} 
	\State $\mathsf{updateState(sq, \emph{leader}, \emph{keyProof},
	\emph{lockProof}, \emph{commitProof})}$

\EndProcedure

\Statex

\Receiving{$n-t$ ``$\textsc{your-view-done}, sq$'' messages}
 
	\State invoke \barrierready{sq}
	\Comment{note that $n-t$ parties must invoke it for
	\barriersync{\emph{sq}} to return}
 
\EndReceiving


\Procedure{exchangeState}{$sq$}

	\State send ``$\textsc{exchange}, sq, \emph{KEY},
	\emph{VALUE}, \emph{COMMIT}$'' to all parties
	\State \textbf{wait} for $n-t$ ``$\textsc{exchange}, sq,*,*$'' messages from different parties

\EndProcedure

\Statex

\Receiving{``$\textsc{exchange}, sq, \emph{key}, \emph{value},
\emph{commit}$''}
 
	\State $\mathsf{check\&updateKey(\emph{key}, \emph{value})}$
	\State $\mathsf{check\&updateCommit(\emph{commit})}$
 
\EndReceiving

\end{algorithmic}
\label{alg:asynch}
\end{algorithm}

\textbf{Termination and communication complexity.}
A naive approach to guarantee termination and avoid an infinite number
of LBV instances in a leader based Byzantine agreement protocols is
to perform a costly communication phase after each LBV instance.
One common approach is to reliably broadcast commit certificates
before halting, while a complementary one is to halt unless
receiving a quorum of complaints from parties that did not decide.
In both cases, the communication cost is $O(t^2)$ even in runs
with at most one failure. 

The key idea of our synchronous protocol is to exploit synchrony in
order to allow honest parties to learn the decision value and at the same time 
help others in a small number of messages.
Instead of complaining (together) after every unsuccessful LBV
instance, each party has its own pre-defined time to
``complain'', in which it learns the up-to-date key and value and
helps others decide via the LBV instance in which it acts as the
leader.

By the conditional progress property and the synchrony assumption,
all honest parties get a commit certificate in LBV instances with
honest leaders. 
Therefore, the termination property is guaranteed since every honest party has its own pre-defined LBV instance, which it invokes only
in case it has not yet decided.
As for the protocol's total communication cost, recall that the
LBV's communication cost is $O(t)$ in the worst case and $0$ in
case its leader already decided and thus does not participate.
In addition, since all honest parties get a commit certificate in the
first LBV instance with an honest leader, we get that the message
cost of all later LBV instances with honest leaders is $0$.
Therefore, the total communication cost of the protocol is
$O(ft+t)$ -- at most $f$ LBVs with Byzantine leaders
and $1$ LBV with an honest one.

\begin{algorithm}[H]
    \caption{Barrier synchronization and Leader-election: protocol for
    a party $p_i$.}
    \label{alg:barrierAndLE}

       
\begin{algorithmic}[1]
\footnotesize

\Statex \textbf{Local variables for Barrier synchronization:}
\State $S_{\emph{barrier}} \gets \{\}$; $\emph{READY} \gets
\emph{false}$

\Statex

\Procedure{barrier-sync}{$sq$}

	\State \textbf{wait} until $\emph{READY} = \emph{true}$

\EndProcedure

\Statex

\Procedure{barrier-ready}{$sq$}

 	\State $\rho \gets \emph{share-sign}_i(\lr{\textsc{shareReady},
 	sq})$ 
 	\State send ``$\textsc{shareReady},sq, \rho$'' to all parties

\EndProcedure

\Statex

\Receiving{``$\textsc{shareReady},sq, \rho$'' from a party $p_j$}
 
	\If{$\emph{share-validate}(\lr{\textsc{shareReady}, sq}, p_j,
	\rho)$}
	
		\State $S_{\emph{barrier}} \gets S_{\emph{barrier}} \cup
		\{\rho\}$
		
		\If{$|S_{\emph{barrier}}| = n-t$}
		
			\State  $\nu \gets \emph{threshold-sign}(S_{\emph{barrierReady}})$
			\State  send ``$\textsc{barrierReady},sq, \nu$'' to all parties

		\EndIf
	
	\EndIf
	
\EndReceiving

\Statex

\Receiving{``$\textsc{barrierReady}, sq, \nu$''}

	\If{$\emph{threshold-validate}(\lr{\textsc{barrierReady}, sq},
	\nu)$}
	
 		\State  send ``$\textsc{barrierReady},sq, \nu$'' to all parties
 		\State $\emph{READY} \gets \emph{true}$ 
	
	\EndIf

\EndReceiving

\algstore{myalg}
\end{algorithmic}

  \begin{algorithmic}[1]
\small
\algrestore{myalg}

\Statex

\Statex \textbf{Local variables for Leader election:}
\State $S_{coin} \gets \{\}$

\Statex

\Procedure{elect}{$sq$}

	\State $\rho \gets \emph{share-sign}_i(sq)$
	\State send ``$\textsc{coinShare},sq,\rho$'' to all parties
	\State \textbf{wait} until $|S_{coin}| = t+1$
	\State  $\nu \gets \emph{threshold-sign}(S_{coin})$
	\State return $p_j$ s.t.\  $j = Hash(\nu) \text{ mod } |\Pi|$
	

\EndProcedure

\Statex

 \Receiving{``$\textsc{coinShare},sq,\rho$'' from $p_j$}
 
 	\If{$\emph{share-validate}(sq,p_j,\rho)$}
 	
 		\State $S_{coin} \gets S_{coin} \cup \{\rho\}$
 	
 	\EndIf
 
 \EndReceiving

\end{algorithmic}

\end{algorithm}

\subsection{Asynchronous fallback}
\label{sub:algAsynch}

In this section, we use the LBV building block to reconstruct
VABA~\cite{VABA}.
Note that achieving an optimal asynchronous protocol is not a
contribution of this paper but reconstructing the VABA protocol with
our LBV building block allows us to safely combine it with
our adaptive synchronous protocol to achieve an optimal optimistic
one.
In addition, we also improve the protocol of VABA in the following
ways: first, parties in VABA~\cite{VABA} never halt, meaning that even
though they decide in expectation in a constant number of
rounds, they operate an unbounded number of them.
We fix it by adding an auxiliary primitive, we call
\emph{help\&tryHalting} in between two consecutive waves (details
below).
Second, VABA guarantees probabilistic termination in all
runs, whereas our version also guarantees standard
termination in eventually synchronous runs.
The full detailed pseudocode of our fallback protocol appears in
Algorithms~\ref{alg:synchauxiliary} , \ref{alg:asynch},  
\ref{alg:barrierAndLE}, and~\ref{alg:halting}.

On a high level, the idea in VABA~\cite{VABA} that was
later generalized in~\cite{ace} is the following: instead of having
a pre-defined leader in every ``round'' of the protocol as most
eventually synchronous protocols and our synchronous protocol have,
they let $n$ leaders operate simultaneously and then randomly choose
one in retrospect.
This mechanism is implemented inside a wave and the agreement
protocol operates in a \emph{wave}-by-\emph{wave} manner s.t.\ 
parties exchange their local states between every two conductive
waves.
To ensure halting, in our version of the protocol, parties also
invoke the \emph{help\&tryHalting} procedure after each wave.
See the \emph{tryPessimistic} procedure in Algorithm~\ref{alg:asynch}
for pseudocode (ignore gray lines at this point) and
Figure~\ref{fig:asynchRun} for an illustration.

\begin{figure*}[th]
  \begin{center}
    \includegraphics[width=0.75\textwidth]{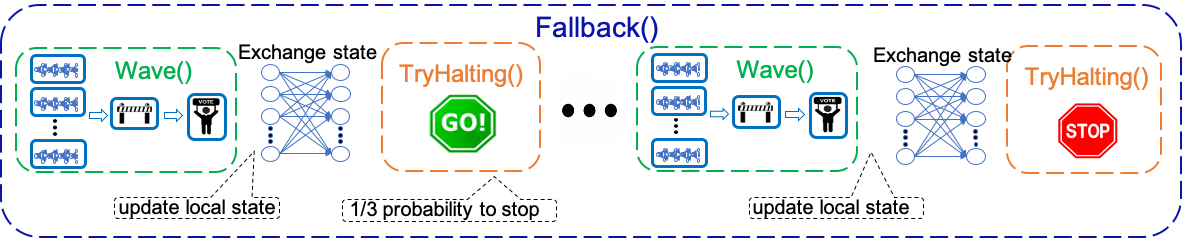}
  \end{center}
  \caption{Asynchronous fallback. Usig linear LBV to reconstruct
  the VABA~\cite{VABA} protocol.}
  \label{fig:asynchRun}
\end{figure*}

\textbf{Wave-by-wave approach.}
To implement the wave mechanism (Algorithm~\ref{alg:asynch}) we use
our LBV and two auxiliary primitives:
Leader-election and Barrier-synchronization
(Algorithm~\ref{alg:barrierAndLE}).
At the beginning of every wave, parties invoke, via
$\mathsf{startView}$, $n$ different LBV instances, each with a
different leader.
Then, parties are blocked in the Barrier-synchronization primitive
until at least $n-2t$ LBV instances 
\emph{complete}.
(An LBV completes when $t+1$ honest parties
get a commit certificate.)
Finally, parties use the Leader-election primitive to elect a unique
LBV instance, wedge it (via $\mathsf{wedgeView}$), and ignore
the rest.
With a probability of $1/3$ parties choose a completed LBV, which
guarantees that after the state exchange phase all honest parties get
a commit certificate, decide, and halt in the
\emph{help\&tryHalting} procedure.
Otherwise, parties update their local state and continue to the next
wave.
An illustration appears in figure~\ref{fig:wave}.

\begin{figure*}[th]
  \begin{center}
    \includegraphics[width=0.8\textwidth]{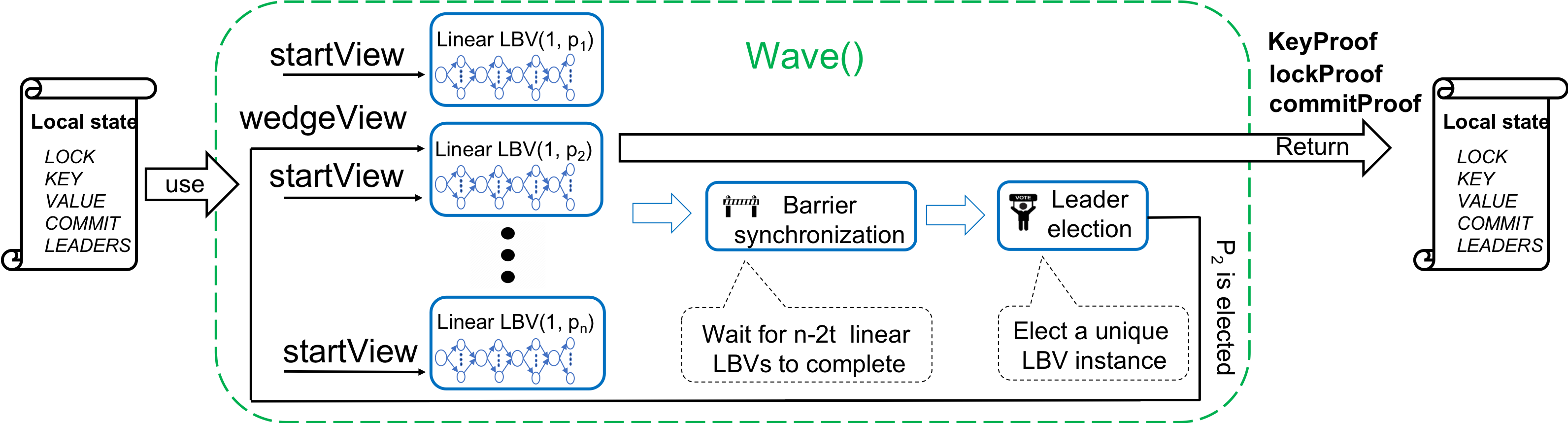}
  \end{center}
  \caption{An illustration of a single wave. The returned \textbf{keyProof},
  \textbf{lockProof}, and \textbf{commitProof} are taken from the
  elected LBV.}
  \label{fig:wave}
\end{figure*}

Since every wave has a probability of $1/3$ to choose a completed LBV
instance, the protocol guarantees probabilistic termination -- in
expectation, all honest parties decide after $3$ waves.
To also satisfy standard termination in eventually
synchronous runs, we ``try synchrony'' after each unsuccessful
wave.
See the gray lines in Algorithm~\ref{alg:asynch}.
Between every two conjunctive waves parties deterministically try to
commit a value in a pre-defined LBV instance.
The preceding \emph{help\&tryHalting} procedure
guarantees that after GST all honest parties invoke
$\mathsf{startView}$ in the pre-defined LBV instance with at most 
$1\Delta$ from each other and thus setting a timeout
to $8\Delta$ is enough for an honest leader to drive progress.
We describe the \emph{help\&tryHalting} procedure in the next section.
The description of the Barrier-synchronization and Leader-election
primitives (Algorithm~\ref{alg:barrierAndLE}) can be found
in~\cite{ace}.

\textbf{Communication complexity.}
The communication cost of the Barrier-synchronization and
Leader-election primitives, as well as that of $n$ LBV instances, is $O(n^2)$,
which brings us to a total of $O(n^2)$ cost per wave.
Since every wave have a probability of $1/3$ to choose a completed
LBV, the protocol operates $3$ waves in expectation.
Therefore, since the communication cost of state exchange and
\emph{help\&tryHalting} is $O(n^2)$, we get
that the total cost, in expectation, is~$O(n^2)$.

\subsection{Optimal optimistic protocol: combine the pieces}
\label{sub:algOptimistic}

\begin{figure*}[h]
  \begin{center}
    \includegraphics[width=0.75\textwidth]{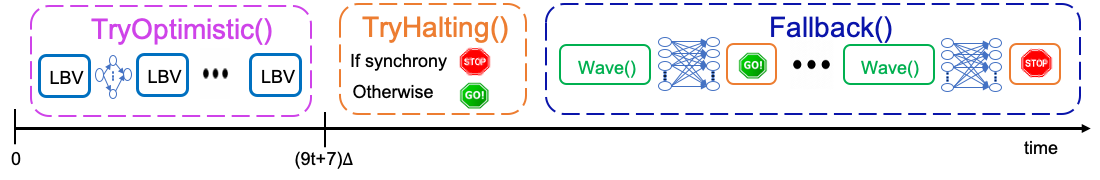}
  \end{center}
  \caption{Illustration of the optimistic protocol. Both parts form
  a sequential composition of LBV instances.}
  \label{fig:optimistic}
\end{figure*}

At a high level, parties first optimistically try the synchronous
protocol (of section~\ref{sub:algSynch}), then invoke
\emph{help\&tryHalting} and continue to the asynchronous
fallback (of section~\ref{sub:algAsynch}) in case a decision has not
been reached.
Pseudocode is given in Algorithm~\ref{alg:optimistic} and an
illustration appears in Figure~\ref{fig:optimistic}.
The parameters passed in Algorithm~\ref{alg:optimistic} 
synchronize the LBV sequence numbers across the different parts of the
protocol.

\begin{algorithm}[h]
\caption{Optimistic byzantine agreement: protocol for a party~$p_i$.} 
\begin{algorithmic}[1]
\footnotesize

\Upon{Optimistic-propose($v_i$)}

	\State $\emph{VALUE} \gets v_i$
	\State $\mathsf{tryOptimistic()}$
	\State $\mathsf{help\&tryHalting(n)}$
	\Comment{Blocking invocation}
	\State $\mathsf{fallback(n)}$

\EndUpon

\end{algorithmic}
\label{alg:optimistic}
\end{algorithm}

One of the biggest challenges in designing an agreement protocol as a
combination of other protocols is to make sure safety is preserved
across them.
Meaning that parties must never decide differently even if they decide
in different parts of the protocol.
In our protocol, however, this is inherently not a concern.
Since both parts use LBV as a building block, we get
safety for free.
That is, if we look at an execution of our protocol in
retrospect, i.e, ignore all LBVs that were not elected in the
asynchronous part.
Then the LBV instances in the synchronous part
together with the elected ones in the asynchronous part form a
sequential composition, which satisfies the Agreement property.

On the other hand, satisfying termination without sacrificing optimal
adaptive complexity is a non-trivial challenge.
Parties start the protocol by optimistically trying the synchronous
part, but unfortunately, at the end of the synchronous part they
cannot distinguish between the case in which the communication was
indeed synchronous and all honest parties decided and the case in
which some honest parties did not decide due to asynchrony.
Moreover, honest parties cannot distinguish between honest
parties that did not decide and thus wish to continue to the
asynchronous fallback part and Byzantine parties that want to move to
the fallback part to increase the communication cost.

To this end, we implement the \emph{help\&tryHalting} procedure, which
stops honest parties from moving to the fallback part in
synchronous runs.
The communication cost of \emph{help\&tryHalting} is $O(ft)$.
The idea is to help parties learn the decision value and move to the
fallback part only when the number of help request indicates that
the run is asynchronous.

\begin{algorithm}[H]
\caption{Help and try halting: Procedure for a party $p_i$.}
\begin{algorithmic}[1]
\footnotesize

\Statex \textbf{Local variables initialization:}
\Statex \hspace*{0.4cm} $S_{\emph{help}} = \{\}$; $\emph{HALT}
\gets true$

\Statex

\Procedure{$\mathsf{help\&tryHalting}$}{$sq$}

	\If{$\emph{COMMIT} = \bot$}
	
		\State $\rho \gets \emph{share-sign}_i(\lr{\textsc{helpRequest},
		sq})$ 
		\State send ``$\textsc{helpRequest},sq, \rho$'' to all parties 
	
	\EndIf
	\State \textbf{wait} until $\emph{HALT} = \emph{false}$

\EndProcedure

\Statex

\Receiving{``$\textsc{helpReply},sq, \emph{commit}$''}

 	\State $\mathsf{check\&updateCommit(\emph{commit})}$

\EndReceiving

\Statex

\Receiving{``$\textsc{helpRequest},sq, \rho$'' from a party $p_j$}
 
	\If{$\emph{share-validate}(\lr{\textsc{helpRequest}, sq}, p_j,
	\rho)$}
	
		\State $S_{\emph{help}} \gets S_{\emph{help}} \cup
		\{\rho\}$
		\State send ``$\textsc{helpReply},sq, \emph{COMMIT}$'' to $p_j$
		
		\If{$|S_{\emph{help}}| = t+1$}
		
			\State  $\nu \gets \emph{threshold-sign}(S_{\emph{help}})$
			\State  send ``$\textsc{complain},sq, \nu$'' to all parties

		\EndIf
	
	\EndIf
	
\EndReceiving

\Statex

\Receiving{``$\textsc{complain}, sq, \nu$''}

	\If{$\emph{threshold-validate}(\lr{\textsc{helpRequst}, sq},
	\nu)$}
	
 		\State  send ``$\textsc{complain},sq, \nu$'' to all parties
 		\State $\emph{HALT} \gets \emph{false}$ 
		 
	
	\EndIf

\EndReceiving

\Statex

\end{algorithmic}
\label{alg:halting}
\end{algorithm}

The pseudocode of \emph{help\&tryHalting} is given in
Algorithm~\ref{alg:halting} and an illustration appears in
Figure~\ref{fig:halting}.
Each honest party that has not yet decided sends a share signed
$\textsc{helpRequest}$ to all other parties.
When an honest party gets an $\textsc{helpRequest}$, the party
replies with its \emph{COMMIT} value, but if it gets $t+1$
$\textsc{helpRequest}$ messages, the party combines the shares to a
threshold signature and sends it in a $\textsc{complain}$ message to
all.
When an honest party gets a $\textsc{complain}$ message for the first
time, it echos the message to all parties and continues to the fallback
part.

\textbf{Termination.}
Consider two cases.
First, the parties move to the fallback part, in which case
(standard) termination is guaranteed in eventually
synchronous runs and probabilistic termination is guaranteed in
asynchronous runs.
Otherwise, less than $t+1$ parties send $\textsc{helpRequest}$ in
\emph{help\&tryHalting}, which implies that at least $t+1$ honest
parties decided and had a commit certificate before invoking
\emph{help\&tryHalting}.
Therefore, all honest parties that did not decide before invoking
\emph{help\&tryHalting} eventually get a $\textsc{helpReply}$ message with a commit certificate and decide as well.

Note that termination does not mean halting. 
In asynchronous runs, $\textsc{helpRequest}$ messages may be arbitrary
delayed and thus parties cannot halt the protocol after deciding in the
synchronous part.
However, it is well known and straightforward to prove that halting
cannot be achieved with $o(t^2)$ communication cost in asynchronous
runs, and thus our protocol is optimal in this aspect.

\begin{figure}
  \centering
  \begin{subfigure}[t]{0.45\textwidth}
    \centering
    \includegraphics[height=1in]{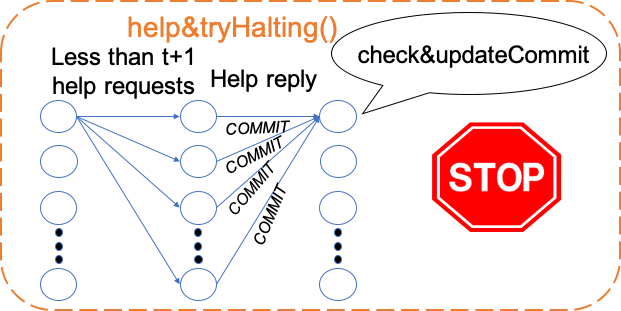}
    \caption{A few $\textsc{helpRequest}$ messages -- help and
    halt.} 
  \end{subfigure}%
   ~
  \begin{subfigure}[t]{0.45\textwidth}
    \centering
    \includegraphics[height=1in]{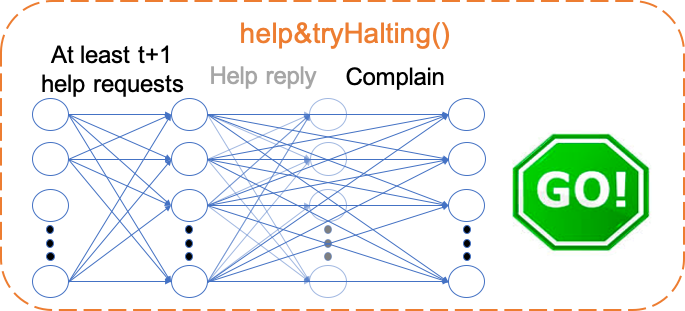}
    \caption{Too much $\textsc{helpRequest}$ messages -- the run
    is asynchronous, move to the fallback part.}
  \end{subfigure}
  \caption{An illustration of the \emph{help\&tryHalting}
  procedure.}
  \label{fig:halting}
\end{figure}

\textbf{Round complexity.}
Since in synchronous runs all parties decide at the end of an LBV
instances with an honest leader, we get that the round complexity in
synchronous runs is $O(f+1)$. Since in asynchronous runs parties may go
though $n$ LBV instances without deciding before starting the fallback,
we get that the round complexity in asynchronous runs is $O(n+1)$ in
expectations.

\textbf{Communuication complexity.}
The synchronous (optimistic) part guarantees that if the run is
indeed synchronous, then all honest parties decide before invoking
\emph{help\&tryHalting}.
The \emph{help\&tryHalting} procedure guarantees that parties continue to
the fallback part only if $t+1$ parties send an
$\textsc{helpRequest}$ message, which implies that they move only if
at least one honest party has not decided in the synchronous part.
Therefore, together they guarantee that honest parties
never move to the fallback part in synchronous runs.

The communication complexity of the synchronous part is $O(ft+t)$, so
to show that the total communication cost of the protocol
in synchronous runs is $O(ft+t)$ we need to show that the cost of
\emph{help\&tryHalting} is $O(ft+t)$ as well.
Since in synchronous runs all honest parties decide in the synchronous part, they do not send $\textsc{helpRequest}$ messages,
and thus no party can send a valid $\textsc{complain}$ message. 
Each Byzantine party that does send $\textsc{helpRequest}$ messages
can cause honest parties to send $O(t)$ replies, which implies a
total communication cost of $O(ft)$ in synchronous runs.

As for all other runs, Theorem~\ref{theorem:ES} 
states that deterministic protocols have an unbounded
communication cost in the worst case.
Thanks to the randomized fallback, our protocol has a communication
cost of $O(t^2)$ in expectation.

 \section{Discussion and Future Directions}
\label{sec:discussion}

In this paper, we propose a new approach to design agreement algorithms
for communication efficient SMR systems.
Instead of designing deterministic protocols for the eventually
synchronous model, which we prove cannot guarantee bounded
communication cost before GST, we propose to design protocols that are
optimized for the synchronous case but also have a randomized fallback
to deal with asynchrony.
Traditionally, most SMR solutions avoid randomized
asynchronous protocols due to their high communication cost.
We, in contrast, argue that this communication cost is reasonable
given that the alternative is an unbounded communication cost
during the wait for eventual synchrony.

We present the first authenticated optimistic protocol with $O(ft+t)$
communication complexity in synchronous runs and $O(t^2)$, in
expectation, in non-synchronous runs.
To strengthen our result, we prove that no deterministic protocol (even
if equipped with perfect cryptographic schemes) can do better in
synchronous runs.
As for the asynchronous runs, the lower bound in\cite{VABA} proves
that $O(t^2)$ is optimal in the worst case of $f=t$.\\

\textbf{Future work.}
Note that our synchronous protocol
satisfies early decision but not early stopping.
That is, all honest parties decide after $O(f)$ rounds, but they
terminate after $O(t)$. 
Therefore, a natural question to ask is whether exist an early stooping
synchronous Byzantine agreement protocol with an optimal adaptive
communication cost.
In addition, it may be possible to improve our protocol's complexity
even further. In particular, the lower bound on communication cost in
synchronous runs applies only to deterministic algorithms, so it might
be possible to circumvent it via randomization~\cite{braud2013fast}.

Another interesting future direction is the question of optimal
resilience in synchronous networks.
Due to the lower bound in~\cite{blum2019synchronous}, the resilience of
our protocol is optimal since the resilience in synchronous runs cannot
be improved as long as the resilience in asynchronous runs is the optimal
$t< n/3$.
However, if we consider synchronous networks in which we do not need to
worry about asynchronous runs, we know that we can tolerate up to $t<
n/2$ failures.
The open question is therefore the following: is there a synchronous Byzantine agreement protocol that tolerates up to $t<
n/2$ failures with an optimal communication complexity of $O(ft+t)$?



\bibliography{bibliography}

\newpage
\appendix

\section{A Correctness Proof of our optimistic byzantine agreement
protocol}
\label{app:upper}

In this section we prove the correctness of our optimistic byzantine
agreement protocol from Section~\ref{sec:algorithms}.

\subsection{Safety.}
We start by proving that our sequential composition of the linear LBV
building block is safe.
In a sequential composition parties sequentially invoke LBV instances 
s.t. for every LBV in the sequence they first invoke
$\mathsf{startView}$ and use the local state variables, then, at some
point they invoke $\mathsf{wedgeView}$ and update the local state with
the returned values and move to the next LBV.

\begin{lemma}
\label{lem:lockonj}

If a party $p$ gets a commit certificate for value $v$ from an LBV
instance with sequence number $sq$, then the local {LOCK} variable
of at least $n-2t$ honest parties is at least $sq$ when they start
the LBV instance with sequence number $sq+1$.

\end{lemma}

\begin{proof}

To generate a commit certificate, a party needs $n-t$
$\textsc{commitShare}$ signatures.
In addition, honest parties wedge an LBV instance before moving to the
next one in sequential compositions.
Thus, at least $n-2t$ honest parties sent their $\textsc{commitShare}$
in the LBV with $sq$ before starting the LBV with $sq+1$. 
The lemma follows since parties set their $\textbf{lockProof}$ 
before sending their $\textsc{commitShare}$ and update 
their \emph{LOCK} variable accordingly before invoking the next LBV.

\end{proof}

\noindent The next corollary follows from Lemma~\ref{lem:lockonj} and
the fact that the local \emph{LOCK} variables are never decreased in
our sequential compositions.

\begin{corollary}
\label{col:lockj}

If a party $p$ gets a commit certificate for value $v$ from an LBV
instance with sequence number $sq$, then there at least $n-2t$ honest
parties which {LOCK} variable is at least $sq$
when they start any LBV with sequence number $sq' > sq$.

\end{corollary}

\begin{lemma}
\label{lem:agreement1}

It is impossible to generate two commit certificates for different
values from the same LBV instance.

\end{lemma}

\begin{proof}

In every LBV instance, to generate a commit certificate for a value
$v$ at least $n-t$ parties need to send a valid $\textsc{commitShare}$
for value $v$.
Therefore, to generate generate two commit certificates for different
values from the same LBV instance we need at least $1$ honest party
need to send two contradicting $\textsc{commitShare}$ message, which
is impossible by the code.

\end{proof}

The next lemma shows that after some party generates a commit
certificate for a value $v$ it is  impossible to generate a valid key
with different value.

\begin{lemma}
\label{lem:keyinduction}

Assume some honest party $p$ gets a commit certificate for value $v$
from an LBV instance with sequence number $sq$.
Than no party can get a valid \textbf{keyProof} on a value other than
$v$ from an LBV instance with sequence number $sq' \geq sq$ in a
sequential composition.

\end{lemma}

\begin{proof}

We prove by induction on LBVs' sequence numbers.

\textbf{Base:} Sequence number $sq$.
To generate a commit certificate for value $v$ at least $n-2t$ honest
parties need to generate a $\textsc{commitShare}$ signature.
An honest party generates a $\textsc{commitShare}$ signature only if
it gets a valid \textsc{lockStep} message with value $v$, which in turn
requires at least $n-2t$ honest parties to generate $\textsc{lockShare}$
signatures on $v$
An honest party generates a $\textsc{lockShare}$ signature on $v$ only
of it gets a valid \textsc{keyStep} message with value $v$, which in
turn requires at least $n-t$ parties to generate
$\textsc{keyShare}$ signatures on $v$.
Thus, if some party gets a commit certificate for value $v$
from an LBV instance with sequence number $sq$, then $n-t$ parties
previously generated $\textsc{keyShare}$ signatures on $v$ in this
LBV instance.
Moreover, since honest parties never generate  $\textsc{keyShare}$
signatures on different values, we get that it is impossible to
generate two valid \textsc{keyStep} messages with different values.
The lemma follows. 

\textbf{step:} Assume the lemma holds for all LBVs with sequence
number $sq''$, $sq \leq sq'' \leq sq'$, we now show that it holds for
$sq' + 1$ as well.
Assume by a way of contradiction that some party
gets a valid \textbf{keyProof} with value $v' \neq v$ from the LBV
with sequence number $sq+1$.
Thus, at least $n-t$ parties generated $\textsc{keyShare}$
signatures on $v'$ in the LBV with $sq+1$.
By Lemma~\ref{lem:lockonj}, there are at least $n-2t$ honest parties
whose local $LOCK \geq sq$ in the LBV with $sq+1$.
Thus, since $n \geq 3t+1$, we get that at least $1$ honest party $p$
whose local $LOCK \geq sq$ generated $\textsc{keyShare}$
signatures on $v'$ in the LBV with $sq+1$.
Therefore, $p$ gets a valid \emph{KEY} for $v'$ with a sequence
number $sq'' \geq sq$.
A contradiction to the inductive assumption.

\end{proof}

\begin{lemma}

Our optimistic byzantine agreement protocol, which is given in
Algorithms~\ref{alg:state}, \ref{alg:LBV-API},
\ref{alg:LBV-messages}, \ref{alg:synch}, \ref{alg:synchauxiliary},
\ref{alg:asynch}, \ref{alg:barrierAndLE}, \ref{alg:optimistic} and
\ref{alg:halting}, satisfies the Agreement property.

\end{lemma}

\begin{proof}

Our protocol sequential composes LBV instances and decides only on
values with a commit certificate.
So we need to show that it is impossible to generate two commit
certificates for different values in a sequential composition of LBV
instances.
Let party $p$ be the first to generate a commit certificate for some
value $v$ and let $sq$ be the sequence number of the LBV instance in
which it was generated.
By lemma~\ref{lem:agreement1}, it is impossible to generate a commit
certificate for a value other then $v$ in the LBV with sequence number
$sq$.
By Lemma~\ref{lem:keyinduction}, no party can get a valid
\textbf{keyProof} on a value other than $v$ from an LBV instance with
sequence number $sq' \geq sq$.
By Lemma~\ref{lem:agreement1}, the local {LOCK} variable
of at least $n-2t$ honest parties is at least $sq$ in any LBV after
the one with sequence number $sq$. 
Therefore, the lemma follows from the fact that at least $n-t$
parties need to contribute signatures in order to generate a
commit certificate and since an honest party whose $LOCK \geq sq$ 
will not generate a $\textsc{keyShare}$ signature on a value $v'$
without getting a valid \emph{KEY} for $v'$ from an LBV with sequence
number $sq' \geq sq$.  

\end{proof}

\subsection{Liveness.}
We now prove that our protocol satisfies termination in all synchronous
and eventually synchronous runs and provide probabilistic termination
in all asynchronous runs.

\begin{lemma}
\label{lem:validKey}

Consider an LBV instance $lbv$ in a sequential composition. 
If some honest party $p$ is locked on a sequence number $sq$ (its
$LOCK = sq$) before starting $lbv$, then at least $n-2t$ honest
parties set their local \emph{KEY} variable with a valid key and
sequence number $sq$ immediately after wedging the LBV instance with
$sq$.

\end{lemma}

\begin{proof}

Since $p$ is locked on $sq$, then it got a valid $\textsc{lockStep}$
message in the LBV instance with sequence number $sq$.
To generate a valid $\textsc{lockStep}$, a party needs $n-t$
$\textsc{lockShare}$ signatures.
Thus, since honest parties first wedge an LBV instance
and then update their local state with the returned values, we get
that at least $n-2t$ honest parties generated a $\textsc{lockShare}$
signature before updating their local \emph{KEY} variable.
Thus, at least $n-2t$ honest parties got a valid $\textbf{keyProof}$
before wedging and thus update their local \emph{KEY} variable
accordingly immediately after wedging. 

\end{proof}

\begin{lemma}
\label{lem:conditionalProgress}

Consider a synchronous run of our protocol, and consider an LBV
instance $lbv$ in the synchronous part, which parties invoke at time
$t$.
If the leader of $lbv$ is honest and it have not decided before
time $t$, then all honest parties decide at time $t+7\Delta$.

\end{lemma}

\begin{proof}

First, by Lemma~\ref{lem:validKey} and since parties overwrite their
local \emph{KEY} variables only with more up-to-date keys, we get 
that at least $n-2t$ honest parties has a \emph{KEY} variable that
unlocks all honest parties (it's sequence number is equal to or higher
than all honest parties' \emph{LOCK}) variable).
By the code, the leader query all parties for their \emph{KEY} and
waits for $n-2t$ replays.
Thus, it gets a reply from at least $1$ honest party that have a key
that unlocks all honest parties.
Therefore, the leader learn this key and thus gets all honest parties
to participates.
The lemma follows from synchrony and the fact that all honest parties
start $lbv$ at the same time and and do not wedge before all honest
parties get al messages.

\end{proof}

\begin{lemma}
\label{lem:synchDecide}

All honest parties decide in all synchronous runs of the protocol. 

\end{lemma}

\begin{proof}

Assume by a way of contradiction that some honest party $p$ does not
decide.
Let $lbv$ be an LBV instance in the synchronous part in which
$p$ is the leader.
By Lemma~\ref{lem:conditionalProgress}, all honest parties decide at
the end of $lbv$.
A contradiction.

\end{proof}

\begin{lemma}
\label{lem:moreThanT1}

If $t+1$ honest parties decide in the synchronous part of a run of our
optimistic protocol, then all honest parties eventually decide.

\end{lemma}

\begin{proof}

By the code of the \emph{help\&tryHalting} procedure, any party $p$
that does not decide in the synchronous part of the protocol sends an
help request to all parties and waits for $n-t$ to reply.
Since $t+1$ honest parties decided in the synchronous part before
invoking \emph{help\&tryHalting}, then $p$ gets a valid commit
certificate and decides as well.

\end{proof}

\begin{lemma}
\label{lem:lessThanT1}

If less than $t+1$ honest parties decide in the synchronous part of a
run of our optimistic protocol, then all honest parties eventually
move to the asynchronous fallback part.

\end{lemma}

\begin{proof}

Since less than $t+1$ honest parties decided before invoking
\emph{help\&tryHalting}, than at least $t+1$ honest parties send an
$\textsc{helpRequest}$ message to all other parties.
Thus all honest parties eventually get $t+1$ help replay, combine them
to a $\textsc{complain}$ message, send it to all other parties, and move to the
fallback part.

\end{proof}

\begin{lemma}
\label{lem:termination}

Our optimistic byzantine agreement protocol, which is given in
Algorithms~\ref{alg:state}, \ref{alg:LBV-API},
\ref{alg:LBV-messages}, \ref{alg:synch}, \ref{alg:synchauxiliary},
\ref{alg:asynch}, \ref{alg:barrierAndLE}, \ref{alg:optimistic} and
\ref{alg:halting}, satisfies termination in all synchronous runs and
provide probabilistic termination in all asynchronous runs.

\end{lemma}

\begin{proof}

Let $r$ be a run of the protocol and consider consider 3 cases:
\begin{itemize}
  
  \item $r$ is synchronous. A (standard) termination is guaranteed by 
  Lemma~\ref{lem:synchDecide}.
  
  \item More than $t+1$ honest parties decide in the synchronous part
  of $r$. The lemma follows from Lemma~\ref{lem:moreThanT1}.
  
  \item Less than $t+1$ honest parties decide in the synchronous part
  of $r$. By Lemma~\ref{lem:lessThanT1}, all honest parties move to
  the asynchronous fallback. The Lemma follows from the termination
  proof in VABA~\cite{VABA} and~\cite{ace}.
  
\end{itemize}

\end{proof}

\subsection{Communication complexity.}

In this section we prove that our protocol has an optimal adaptive 
synchronous communication complexity and and optimal worst case
asynchronous communication.

\begin{lemma}
\label{lem:synchPartCost}

The communication cost of the synchronous part in synchronous runs of
our optimistic protocol is $O(ft + t)$.

\end{lemma}

\begin{proof}

Consider a synchronous run.
The communication cost of an LBV instance (with byzantine or honest
leader) plus the leader-to-all all-to-leader key learning phase is at
most $O(t)$ and the communication cost of an LBV instance with an
honest leader that does not drive progress since it has already
decided before is $0$ (honest parties only reply to leaders
messages).
By Lemma~\ref{lem:conditionalProgress} all honest parties decide in
the first LBV instance with an honest leader that drive progress.
Therefore, there is at most $1$ honest leader that drive progress in
the LBV in which it acts as the leader.
Thus, the total communication cost of all LBVs with honest leaders is
$O(ft)$.
Hence, since every byzantine leader can make honest parties pay at
most $O(t)$ communication cost in the LBV instance in which the
byzantine party is the leader, we get to a total communication cost
of $O(ft +t)$.

\end{proof}

\begin{lemma}
\label{lem:haltingCost}

The communication cost of the \emph{help\&tryHalting} procedure
in synchronous runs of our optimistic protocol is $O(ft)$.

\end{lemma}

\begin{proof}

Consider a synchronous run $r$.
By Lemma~\ref{lem:synchDecide}, all honest parties decide in the
synchronous part of $r$.
Thus, no honest party sends an $\textsc{helpRequest}$ message in the 
\emph{help\&tryHalting} procedure and it is impossible to generate a
valid $\textsc{complain}$ message.
An $\textsc{helpRequest}$ by a byzantine party causes all honest party
to reply, which cost $O(t)$ in communication cost.
Therefore, the total communication cost of the \emph{help\&tryHalting}
procedure in synchronous runs is $O(ft)$.

\end{proof}

\begin{lemma}
\label{lem:synchCost}

The synchronous communication cost of our optimistic byzantine
agreement protocol is $O(ft + t)$.

\end{lemma}

\begin{proof}

By Lemma~\ref{lem:synchDecide}, all honest parties decide in the
synchronous part of $r$ and
thus no honest party sends an $\textsc{helpRequest}$ message in the 
\emph{help\&tryHalting} procedure.
Therefore, it is impossible to generate a valid $\textsc{complain}$
message, and thus no honest party moves to the fallback part.
The lemmas follows from Lemmas~\ref{lem:synchPartCost}
and~\ref{lem:haltingCost}.

\end{proof}

\begin{lemma}
\label{lem:asynchCost}

The asynchronous communication cost of our optimistic byzantine
agreement protocol is $O(t^2)$ in the worst case.

\end{lemma}

\begin{proof}

Consider an asynchronous run $r$.
We first prove that the worst case communication cost of the
synchronous part and the \emph{help\&tryHalting}
procedure is $O(t^2)$:
\begin{itemize}
  
  \item Synchronous part. The synchronous part consists of $n$ LBV
  instances with a one-to-all and all-to-one communication phase in
  between. Since the communication cost of the LBV building block is
  at most $O(t)$, we get that the total communication cost of the
  synchronous part is $O(nt) = O(t^2)$.
  
  \item The\emph{help\&tryHalting} procedure. Every honest party sends
  at most one $\textsc{helpRequest}$ message, $\textsc{helpReplay}$
  message, and $\textsc{complain}$ message.
  Therefore, since each of the messages contains a constant number of
  words, we get that the worst case communication complexity of the
  \emph{help\&tryHalting} procedure is $O(t^2)$.
  
\end{itemize}

The Lemma follows from the communication cost analysis of the fallback
algorithm, which appears in VABA~\cite{VABA}.

\end{proof}

The next corollary follows directly from Lemmas~\ref{lem:synchCost}
and~\ref{lem:asynchCost}:

\begin{corollary}

The adaptive synchronous and worst case asynchronous communication
cost of our optimistic byzantine agreement protocol, which is given
in Algorithms~\ref{alg:state}, \ref{alg:LBV-API}, \ref{alg:LBV-messages}, \ref{alg:synch},
\ref{alg:synchauxiliary}, \ref{alg:asynch}, \ref{alg:barrierAndLE},
\ref{alg:optimistic} and \ref{alg:halting} is $O(ft + t)$ and
$O(t^2)$, respectively.

\end{corollary}

\section{Threshold signatures}
\label{app:TS}

At the beginning of every execution, every party $p_i$ gets a private
function $\emph{share-sign}_i(m)$ from the dealer, which gets a
message $m$ and returns a
signature-share $\sigma_i$.
In addition, every party gets the following functions: (1)
$\emph{share-validate}(m,i,\sigma_i)$, which gets a message $m$, a party
identification $i$, and a signature-share $\sigma_i$, and returns
\emph{true} or \emph{false}; (2)
$\emph{threshold-sign}(\Sigma)$, which gets a set of signature-shares
$\Sigma$, and returns a threshold signature $\sigma$; and (3)
$\emph{threshold-validate}(m,\sigma)$, which gets a message $m$ and a
threshold signature $\sigma$, and returns \emph{true} or \emph{false}.
We assume that the above functions satisfy the following properties:

\begin{itemize}

%
  \item \emph{Share validation:} For all $i$, $1 \leq i \leq n$ and
  for every messages $m$, (1) $\emph{share-validate}(m,i,\sigma) = true$ if and only if $\sigma =
  \emph{share-sign}_i(m)$, and (2) if $p_i$ is
  honest, then it is infeasible for the adversary to compute
  $\emph{share-sign}_i(m)$.


   \item \emph{Threshold validation:} For every message $m$,
  $\emph{threshold-validate}(m,\sigma) = true$ if and only if
  $\sigma = \emph{threshold-sign}(\Sigma)$ s.t.\
  $|\Sigma| \geq n-t$ and for every $\sigma_i \in \Sigma$ there is a
  party $p_i$ s.t.\ $\emph{share-validate}(m,i,\sigma) = true$.


\end{itemize}

\end{document}